\documentclass[letter,11pt,final]{article}
\usepackage[numbers,sort]{natbib}
\usepackage{typearea}
\typearea{14}

\usepackage{authblk}
\usepackage{comment}
\usepackage[T1]{fontenc}
\usepackage[noend]{algorithmic}
\usepackage{algorithm}
\usepackage{multirow}
\usepackage{amssymb}

\usepackage{amsmath}

\usepackage{amsthm,thm-restate}
\usepackage{todonotes}
\usepackage{enumerate}
\usepackage[shortlabels]{enumitem}
\usepackage{dsfont}
\usepackage[colorlinks=true,citecolor=blue,urlcolor=blue,bookmarks=false]{hyperref}
\usepackage{soul}  
\usepackage{cleveref} 
\usepackage{footnote}
\usepackage{subcaption}
\usepackage{xspace}
\usepackage{nicefrac}
\usepackage{tikz}
\usetikzlibrary{decorations.pathmorphing, positioning, decorations.markings}
\usetikzlibrary[topaths]
\usepackage{ragged2e} 
\tikzstyle{replaced edge} = [line width=1.5pt,blue!80]
\tikzstyle{vertex}= [draw,circle,inner sep=0.1cm]
\tikzstyle{BigLetterVertex}= [draw,circle,inner sep=0.075cm]
\usepackage{framed}
\usepackage[framemethod=tikz]{mdframed}

\theoremstyle{plain}
\newtheorem{thm}{Theorem}[section]
\newtheorem{cor}[thm]{Corollary}

\newtheorem{fact}[thm]{Fact}
\newtheorem{lem}[thm]{Lemma}

\newtheorem{Def}[thm]{Definition}
\newtheorem{obs}[thm]{Observation}

\newcommand{\eps}{\varepsilon}

\newcommand{\calS}{\mathcal{S}}

\newcommand{\calB}{\mathcal{B}}
\newcommand{\calP}{\mathcal{P}}
\newcommand{\calT}{\mathcal{T}}
\newcommand{\calM}{\mathcal{M}}
\newcommand{\calD}{\mathcal{D}}
\newcommand{\tO}{\tilde{O}}
\newcommand{\tOmega}{\tilde{\Omega}}
\newcommand{\mc}[1]{\mathcal{#1}}
\newcommand{\reals}{\mathbb{R}}
\newcommand{\integers}{\mathbb{Z}}

\newcommand{\cev}[1]{\reflectbox{\ensuremath{\vec{\reflectbox{\ensuremath{#1}}}}}}
\newcommand{\nth}{\textsuperscript{th}\xspace}

\newcommand{\defeq}{\triangleq}
\newcommand{\primal}{\textsc{ConcurrentFlow}$_\mc M$}
\newcommand{\dual}{\textsc{Cut}$_\mc M$}
\newcommand{\pathto}{{\rightsquigarrow}}
\newcommand{\tl}{{\tilde{\ell}}}
\newcommand{\tI}{{\tilde{I}}}
\newcommand{\at}[1]{^{(#1)}} 

\newcommand{\diam}{{\mathrm{diam}}}
\newcommand{\inst}{{I}}

\newcommand{\agap}{{\alpha_{\mathrm{gap}}}}

\title{Network Coding Gaps for Completion Times of Multiple Unicasts}
\author{Bernhard Haeupler}
\author{David Wajc}
\author{Goran Zuzic}
\affil{Carnegie Mellon University\\
	\{haeupler, dwajc, gzuzic\}@cs.cmu.edu}
\date{\vspace{-1cm}}

\begin{document}
\maketitle

\pagenumbering{gobble}
	\footnotetext{Supported in part by NSF grants CCF-1618280, CCF-1814603, CCF-1527110, NSF CAREER award CCF-1750808, and a Sloan Research Fellowship, as well as a DFINITY scholarship.
	}

\begin{abstract}
We study network coding gaps for the problem of makespan minimization of multiple unicasts. In this problem distinct packets at different nodes in a network need to be delivered to a destination specific to each packet, as fast as possible. The \emph{network coding gap} specifies how much coding packets together in a network can help compared to the more natural approach of routing.

\medskip

While makespan minimization using routing has been intensely studied for the multiple unicasts problem, no bounds on network coding gaps for this problem are known. We develop new techniques which allow us to upper bound the network coding gap for the makespan of $k$ unicasts, proving this gap is at most polylogarithmic in $k$. Complementing this result, we show there exist instances of $k$ unicasts for which this coding gap is polylogarithmic in $k$. Our results also hold for average completion time, and more generally any $\ell_p$ norm of completion times.
\end{abstract}


%
%
%
%
%
%

\newpage
\pagenumbering{arabic}

\section{Introduction}\label{sec:intro}

In this paper we study the natural mathematical abstraction of what is arguably the most common network communication problem: \emph{multiple unicasts}. 
In this problem, distinct packets of different size are at different nodes in a network, and each packet needs to be delivered to a specific destination as fast as possible. That is, minimizing the \emph{makespan}, or the time until all packets are delivered.

\smallskip

All known multiple-unicast solutions employ (fractional) \emph{routing} (also known as store-and-forward protocols), i.e., network nodes potentially subdivide packets and route (sub-)packets to their destination via store and forward operations, while limited by edge capacities.
The problem of makespan minimization of routing has been widely studied over the years. A long line of work \cite{rothvoss2013simpler,peis2011universal,srinivasan2001constant,bertsimas1999asymptotically,koch2009real,busch2004direct,peis2009packet,rabani1996distributed,ostrovsky1997universal,auf1999shortest,leighton1999fast,leighton1994packet,scheideler2006universal}, starting with the seminal work of Leighton, Maggs, and Rao \cite{leighton1994packet}, studies makespan minimization for routing along fixed paths. The study of makespan minimization for routing (with the freedom to pick paths along which to route) resulted in approximately-optimal routing, first for asymptotically-large packet sizes \cite{bertsimas1999asymptotically}, and then for all packet sizes \cite{srinivasan2001constant}. 
%
%

\smallskip

It seems obvious at first that routing packets, as though they were physical commodities, is the only way to solve network communication problems, such as multiple unicasts. Surprisingly, however, results discovered in the 2000s~\cite{ahlswede2000network} suggest that information need not flow through a network like a physical commodity. For example, nodes might not just forward information, but instead send out XORs of received packets. Multiple such XORs or linear combinations can then be recombined at destinations to reconstruct any desired packets. An instructive example is to look at the XOR $C \oplus M$ of two $s$-bit packets, $C$ and $M$. While it is also $s$ bits long, one can use it to reconstruct either all $s$ bits of $C$ or all $s$ bits of $M$, as long as the other packet is given. Such network coding operations are tremendously useful for network communication problems, but they do not have a physical equivalent. Indeed, the $C \oplus M$ packet would correspond to some $s$ ounces of a magic ``caf\'e latte'' liquid  with the property that one can extract either $s$ ounces of milk or $s$ ounces of coffee from it, as long as one has enough of the other liquid already. 
Over the last two decades, many results demonstrating gaps between the power of network coding and routing have been published~(e.g., \cite{ahlswede2000network,wu2005minimum,harvey2006capacity,li2004network,xiahou2014geometric,langberg2009multiple,al2008k,deb2006algebraic,goel2008network,fragouli2008efficient,chekuri2015delay,wang2016sending,haeupler2015simple,haeupler2016analyzing,katti2006xors}). Attempts to build a comprehensive theory explaining what is or is not achievable by going beyond routing have given rise to an entire research area called network information theory.

\smallskip

The question asked in this paper is: 
\begin{center}
	\begin{minipage}{0.93\linewidth}
		\begin{mdframed}[hidealllines=true, backgroundcolor=gray!20]
			``How much faster than routing can network coding be for any multiple-unicast instance?'' 
		\end{mdframed}
	\end{minipage}
\end{center}
In other words, what is the (multiplicative) \emph{network coding gap} for makespan of multiple unicasts. Surprisingly, no general makespan coding gap bounds were known prior to this work. This is in spite of the vast amount of effort invested in understanding routing strategies for this problem, and ample evidence of the benefits of network coding. 

\smallskip 
This question was studied in depth for the special case of \emph{asymptotically-large} packet sizes, otherwise known as \emph{throughput} maximization (e.g., \cite{li2004network,harvey2004comparing,kramer2006edge,adler2006capacity,jain2006capacity,harvey2006capacity,yin2018reduction,li2004network,xiahou2014geometric,langberg2009multiple,al2008k}). Here, the maximum \emph{throughput} of a multiple-unicast instance can be defined as $\sup_{w \rightarrow \infty} w / C(w)$, where $C(w)$ is the makespan of the fastest protocol for the instance after increasing all packet sizes by a factor of $w$ (see \Cref{sec:completion-vs-throughput}). In the throughput setting, no instances are known where coding offers \emph{any} advantage over routing, and this is famously conjectured to be the case for all instances \cite{li2004network,harvey2004comparing}. 
This conjecture, if true, has been proven to have surprising connections to various lower bounds \cite{afshani2019lower,adler2006capacity,farhadi2019lower}. Moreover, by the work \citet{afshani2019lower}, a throughput coding gap of $o(\log k)$ for all multiple-unicast instances with $k$ unicast pairs (\emph{$k$-unicast instances}, for short) would imply explicit \emph{super-linear circuit lower bounds}---a major breakthrough in complexity theory. Such a result is currently out of reach, as the best known upper bound on throughput coding gaps is $O(\log k)$, which follows easily from the same bound on multicommodity flow/sparsest cut gaps \cite{aumann1998log,linial1995geometry}.

\smallskip

In this work we prove makespan coding gaps for the general problem of \emph{arbitrary} packet sizes. In particular, we show that this gap is at most $O(\log^2 k)$ for any $k$-unicast instance (for the most interesting case of similar-sized packet sizes). We note that any coding gap upper bound for this more general setting immediately implies the same bound in the throughput setting (\Cref{sec:completion-vs-throughput}), making our general bound only quadratically larger than the best known bound for the special case of throughput. Complementing our results, we prove that there exist $k$-unicast instances where the network coding gap is $\Omega(\log^c k)$ for some constant $c > 0$. 

\smallskip

To achieve our results we develop novel techniques that might be of independent interest. The need for such new tools is due to makespan minimization for general packet sizes needing to take both source-sink distances as well as congestion issues into account. This is in contrast with the throughput setting, where bounds must only account for congestion, since asymptotically-large packet sizes make distance considerations inconsequential. 
For our more general problem, we must therefore develop approaches that are both congestion- and distance-aware.
One such approach is given by a new combinatorial object we introduce, dubbed the \emph{moving cut}, which allows us to provide a \emph{universally optimal} characterization of the coding makespan. That is, it allows us to obtain tight bounds (up to polylog terms) on the makespan of \emph{any} given multiple-unicast instance. We note that moving cuts can be seen as generalization of prior approaches that were (implicitly) used to prove unconditional lower bounds in distributed computing on specially crafted networks~\citep{dassarma2012distributed,peleg2000near}; the fact they provide a \emph{characterization} on all networks and instances is novel. This underlies our main result---a polylogarithmic upper bound on the makespan coding gap for any multiple-unicast instance.

\subsection{Preliminaries}\label{sec:prelims}
In this section we define the completion-time communication model. We defer the, slightly more general, information-theoretic formalization to \Cref{sec:network-coding-model}.

A \emph{multiple-unicast instance} $\calM=(G,\calS)$ is defined over a communication network, represented by an connected undirected graph $G = (V,E)$ with capacity $c_e\in \integers_{\ge 1}$ for each edge $e$. The $k\defeq|\calS|$ \emph{sessions} of $\calM$ are denoted by $\calS = \{ (s_i, t_i, q_i) \}_{i=1}^k$. Each session consists of source node $s_i$, which wants to transmit a packet to the sink $t_i$, consisting of $q_i \in \integers_{\ge 1}$ sub-packets. Without loss of generality we assume that a uniform sub-packetization is used; i.e., all sub-packets have the same size (think of sub-packets as the underlying data type, e.g., field elements or bits). For brevity, we refer to an instance with $k$ sessions as a \emph{$k$-unicast} instance.

A \emph{protocol} for a multiple-unicast instance is conducted over finitely-many \emph{synchronous time steps}. 
Initially, each source $s_i$ knows its packet, consisting of $d_i$ sub-packets.
At any time step, the protocol instructs each node $v$ to send a different packet along each of its edges $e$. The packet contents are computed with some predetermined function of packets received in prior rounds by $v$ or originating at $v$.
\emph{Network coding protocols} are unrestricted protocols, allowing each node to send out any function of the packets it has received so far.
On the other hand, \emph{routing protocols} are a restricted, only allowing a node to forward sub-packets which it has received so far or that originate at this node. 

We say a protocol for multiple-unicast instance has \emph{completion times} $(T_1,T_2,\dots,T_k)$ if for each $i\in [k]$, after $T_i$ time steps of the protocol the sink $t_i$ can determine the $d_i$-sized packet of its source $s_i$. The complexity
of a protocol is determined by functions $\mc{C} : \mathbb{R}_{\ge 0}^k \to \mathbb{R}_{\ge 0}$ of its completion times. For example, a protocol with completion times $(T_1,T_2,\dots,T_k)$ has \emph{makespan} $\max_{i\in [k]} T_i$ and \emph{average completion time} $(\sum_{i\in [k]}T_i)/k$. Minimizing these measures is a special case of minimizing weighted $\ell_p$ norms of completion time, namely minimizing $(\sum_{i\in [k]} w_i\cdot T_i^p)^{1/p}$ for some $\vec{w}\in \mathbb{R}^k$ and $p\in \mathbb{R}_{\geq 0}$.

%
%

Since coding protocols subsume routing ones, for any function $\mc{C}$ of completion times, and for any multiple-unicast instance, the fastest routing protocol is no faster than the fastest coding protocol. Completion-time coding gaps characterize how much faster the latter is.

\begin{Def}\label{def:coding-gap}(Completion-time coding gaps) For any function $\mc{C}:\mathbb{R}_{\ge 0}^k \rightarrow \mathbb{R}_{\ge 0}$ of completion times, the \emph{network coding gap for $\mc{C}$} for a $k$-unicast instance $\calM = (G,\calS)$ is the ratio of the smallest $\mc{C}$-value of any routing protocol for $\calM$ and the smallest $\mc{C}$-value of any network coding protocol for $\calM$.
\end{Def}

	We note that the multiple-unicast instance problem can be further generalized, so that each edge has both capacity and \emph{delay}, corresponding to the amount of time needed to traverse the edge. This more general problem can be captured by replacing each edge $e$ with a path with unit delays of total length proportional to $e$'s delay. 
	As we show, despite path length being crucially important in characterizing completion times for multiple-unicast instances, this transformation does not affect the worst-case coding gaps, which  are independent of the network size (including after this transformation). We therefore consider only unit-time delays in this paper, without loss of generality.

\subsection{Our Contributions}\label{sec:results}

In this work we show that completion-time coding gaps of multiple unicasts are vastly different from their throughput counterparts, which are conjectured to be trivial (i.e., equal to one). For example, while  
the throughput coding gap is always one for instances with $k=2$ sessions \cite{hu1963multi}, for makespan it is easy to derive instances with $k=2$ sessions and coding gap of $\nicefrac{4}{3}$ (based on the butterfly network).
Having observed that makespan coding gaps can in fact be nontrivial, we proceed to study the potential asymptotic growth of such coding gaps as the network parameters grow.
We show that the makespan coding gap of multiple unicasts with $k$ sessions and packet sizes $\{d_i\}_{i\in [k]}$ is polylogarithmic in the problem parameters, $k$ and $\sum_i d_i / \min_i d_i$, but independent of the network size, $n$. 
The positive part of this result is given by the following theorem.

\begin{center}
	\begin{minipage}{\linewidth}
		\begin{mdframed}[hidealllines=true, backgroundcolor=gray!20]
		\begin{restatable}{thm}{latencyub}\label{wc-gap-ub}
			The network coding gap for makespan of any $k$-unicast instance is at most 
			$$O\left(\log(k) \cdot \log \left(\sum_i d_i / \min_i d_i\right)\right).$$
		\end{restatable}
		\end{mdframed}
	\end{minipage}
\end{center}
For similarly-sized packets, this bound simplifies to $O(\log^2k)$. For different-sized packets, our proofs and ideas in \cite{plotkin1995improved} imply a coding gap of $O(\log k \cdot \log(nk))$. 
Moreover, our proofs are constructive, yielding for any $k$-unicast instance $\calM$ a routing protocol which is at most $O(\log k \cdot \log (\sum_i d_i / \min_i d_i))$ and $O(\log k\cdot \log (nk))$ times slower than the fastest protocol (of any kind) for $\calM$. We note that our upper bounds imply the same upper bounds for throughput (see \Cref{sec:completion-vs-throughput}). Our bounds thus also nearly match the best coding gap of $O(\log k)$ known for this special case of makespan minimization.

\medskip

On the other hand, we prove that a polylogarithmic gap as in \Cref{wc-gap-ub} is inherent, by providing an infinite family of multiple-unicast instances with unit-sized packets ($d_i=1$ for all $i\in [k]$) exhibiting a polylogarithmic makespan coding gap. 
\begin{center}
	\begin{minipage}{\linewidth}
		\begin{mdframed}[hidealllines=true, backgroundcolor=gray!20]
			\begin{restatable}{thm}{latencylb}\label{wc-gap-lb}
				There exists an absolute constant $c>0$ and 
				an infinite family of $k$-unicast instances whose makespan coding gap is 
				at least $\Omega(\log^{c} k).$
			\end{restatable}
		\end{mdframed}
	\end{minipage}
\end{center}

Building on our results for makespan we obtain similar results to Theorems \ref{wc-gap-ub} and \ref{wc-gap-lb} for average completion time and more generally for any weighted $\ell_p$ norm of completion times. 

\subsection{Techniques}\label{sec:techniques}

Here we outline the challenges faced and key ideas needed to obtain our results, focusing on makespan.

\subsubsection{Upper Bounding the Coding Gap}
As we wish to bound the ratio between the best makespan of any routing protocol and any coding protocol, we need both upper and lower bounds for these best makespans. As it turns out, upper bounding the best makespan is somewhat easier.
The major technical challenge, and our main contribution, is in deriving lower bounds on the optimal makespan of any given multiple-unicast instance. Most notably, we formalize a technique we refer to as the \emph{moving cut}. Essentially the same technique was used to prove that distributed verification is hard on one \emph{particular} graph that was designed specifically with this technique in mind~\cite{elkin2006unconditional,dassarma2012distributed,peleg2000near}. Strikingly, we show that the moving cut technique gives an almost-tight characterization (up to polylog factors) of the coding makespan for \emph{every} multiple-unicast instance (i.e., it gives \emph{universally} optimal bounds).

We start by considering several prospective techniques to prove that no protocol can solve an instance in fewer than $T$ rounds, and build our way up to the moving cut. For any multiple-unicast instance, $\max_{i \in [k]} \textrm{dist}(s_i, t_i)$, the maximum distance between any source-sink pair, clearly lower bounds the coding makespan. However, this lower bound can be arbitrarily bad since it does not take edge congestion into account; for example, if all source-sink paths pass through one common edge. Similarly, any approach that looks at sparsest cuts in a graph is also bound to fail since it does not take the source-sink distances into account.

Attempting to interpolate between both bounds, one can try to extend this idea by noting that a graph that is ``close'' (in the sense of few deleted edges) to another graph with large source-sink distances must have large makespan for routing protocols. For simplicity, we focus on instances where all capacities and demands are one, i.e., $c_e = 1$ for every edge $e$ and $d_i = 1$ for all $i$, 
which we refer to as \emph{simple} instances. The following simple lemma illustrates such an approach.
\begin{lem}\label{obs:cut-implies-routing-lb}
	Let $\calM = (G,\calS)$ be a simple $k$-unicast instance. Suppose that after deleting some edges $F \subseteq E$, any sink is at distance at least $T$ from its source; i.e., $\forall i\in [k]$ we have $\textrm{dist}_{G\setminus F}(s_i, t_i) \ge T$. Then any routing protocol for $\calM$ has makespan at least $\min\left\{T, k/|F|\right\}$.
\end{lem}

\begin{proof}
	For any sets of flow paths between all sinks and source, either (1) all source-sink flow paths contain at least one edge from $F$, incurring a congestion of $k/|F|$ on at least one of these $|F|$ edges, or (2) there is a path not containing any edge from $F$, hence having a hop-length of at least $T$. 
	Either way, any routing protocol must take at least $\min\{T,k/|F|\}$ to route along these paths.
\end{proof}

Perhaps surprisingly, the above bound does not apply to general (i.e., coding) protocols. Consider the instance in \Cref{fig:gap-53}. There, removing the single edge $\{S,T\}$ increases the distance between any source-sink pair to $5$, implying any routing protocol's makespan is at least $5$ on this instance. However, there exists a network coding protocol with makespan $3$: Each source $s_i$ sends its input to its neighbor $S$ and all sinks $t_j$ for $i \neq j$ along the direct 3-hop path $s_i- - -t_j$. Node $S$ computes the XOR of all inputs, passes this XOR to $T$ who, in turn, passes this XOR to all sinks $t_j$, allowing each sink $t_j$ to recover its source $s_j$'s packets by canceling all other terms in the XOR.

\def\k{5}%
\begin{figure}[h]
	\begin{center}
		\begin{tikzpicture}[scale=0.9,
		vtx/.style={draw, circle},
		decoration={markings,mark=between positions .33 and 0.97 step .33 with {\node[draw, circle, scale=0.3] {};}}
		]
		\node[vtx] at (2*\k+4,2) (S) {$S$};
		\node[vtx] at (2*\k+4,0) (T) {$T$};
		\draw (S) -- (T);
		\foreach \x in {1,...,\k} {
			\node[vtx] at (2*\x,2) (s\x) {$s_\x$};
			\path (s\x) edge[-,bend left=17,blue] (S);
			\node[vtx] at (2*\x,0) (t\x) {$t_\x$};
			\path (T) edge[-,bend left=17,blue] (t\x);	
		}
		\foreach \x in {1,...,\k} {
			\foreach \y in {1,...,\k} {
				\ifthenelse{\NOT \x = \y}{\draw[line width=1.25pt] (s\x) -- (t\y);}{}
			}
		}
		\end{tikzpicture}
		\centering
		\vspace{-0.2cm}
		\caption{A family of instances with $k=\k$ pairs of terminals and makespan coding gap of $\nicefrac{5}{3}$. Thick edges represent paths of 3 hops, while thin (black and blue) edges represent single edges. In other words, each of the $k$ sources $s_i$ has a path of $3$ hops (in black and bold) connecting it to every sink $t_j$ for all $j\neq i$. Moreover, all sources $s_i$ neighbor a node $S$, which also neighbors node $T$, which neighbors all sinks $t_j$.}
		\label{fig:gap-53}
		\vspace{-0.4cm}
	\end{center}	
\end{figure}
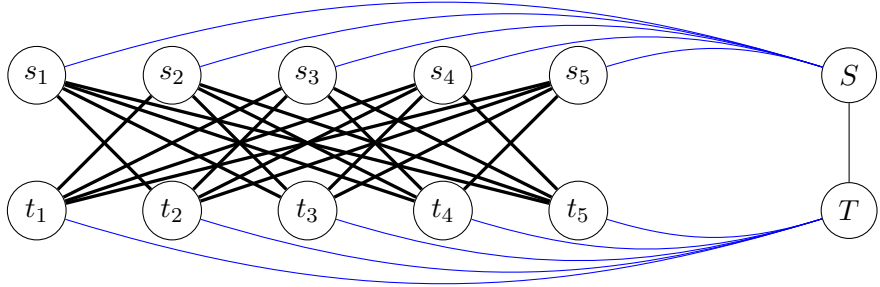

One can still recover a valid general (i.e., coding) lower bound by an appropriate strengthening of \Cref{obs:cut-implies-routing-lb}: one has to require that \emph{all} sources be far from \emph{all} sinks in the edge-deleted graph. This contrast serves as a good mental model for the differences between coding and routing protocols.

\begin{lem}\label{obs:cut-implies-coding-lb}
	Let $\calM = (G,\calS)$ be a simple $k$-unicast instance. Suppose that after deleting some edges $F \subseteq E$, any sink is at distance at least $T$ from \textbf{any} source; i.e., $\forall \pmb{i},\pmb{j}\in [k]$ we have $\textrm{dist}_{G\setminus F}(\pmb{s_i}, \pmb{t_j}) \ge T$. Then any (network coding) protocol for $\calM$ has makespan at least $\min\left\{T, k/|F|\right\}$.
\end{lem}

\begin{proof}
	We can assume all sources can share information among themselves for free (e.g., via a common controlling entity) since this makes the multiple-unicast instance strictly easier to solve; similarly, suppose that the sinks can also share information. Suppose that some coding protocol has makespan $T' < T$. Then all information shared between the sources and the sinks has to pass through some edge in $F$ at some point during the protocol. However, these edges can pass a total of $|F| \cdot T'$ packets of information, which has to be sufficient for the total of $k$ source packets. Therefore, $|F| \cdot T' \ge k$, which can be rewritten as $T' \ge k/|F|$. The makespan is therefore at least $T'\geq \min\{T,k/|F|\}$.
\end{proof}

Unfortunately, \Cref{obs:cut-implies-coding-lb} is not always tight and it is instructive to understand when this happens. One key example is the previously-mentioned instance studied in the influential distributed computing papers \cite{peleg2000near,elkin2006unconditional,dassarma2012distributed} (described in \Cref{fig:dist-graph}), where congestion and dilation both play key roles. Informally, this network was constructed precisely to give an $\tilde{\Omega}(\sqrt{n})$ makespan lower bound (leading to the pervasive $\tilde{\Omega}(\sqrt{n}+D)$ lower bound for many global problems in distributed computing~\cite{dassarma2012distributed}). The intuitive way to explain the $\tilde{\Omega}(\sqrt{n})$ lower bound is to say that one either has to communicate along a path of length $\sqrt{n}$ or \emph{all} information needs to shortcut significant distance over the tree, which forces all information to pass through near the top of the tree, implying congestion of $\tOmega(\sqrt n)$. \Cref{obs:cut-implies-coding-lb}, however, can at best certify a lower bound of $\tilde{\Omega}(n^{1/4})$ for this instance. That is, this lemma's (coding) makespan lower bound can be \emph{polynomially} far from the optimal coding protocol's makespan. 

\vspace{-0.3cm}
\begin{figure}[h]
	\begin{center}
		\includegraphics[scale=0.5]{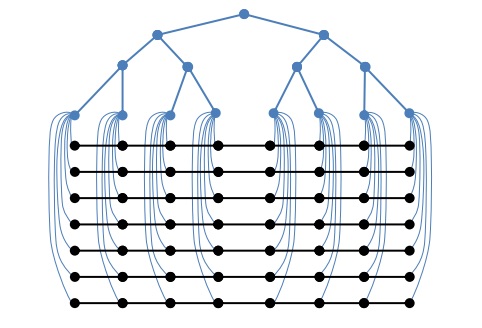}
		\centering
		\vspace{-0.3cm}
		\caption{The hard instance for distributed graph problems \cite{elkin2006unconditional,dassarma2012distributed,peleg2000near}, as appears in \cite{ghaffari2016distributed}. The multiple-unicast instance has $\Theta(n)$ nodes and is composed of $\sqrt{n}$ disjoint paths of length $\sqrt{n}$ and a perfectly balanced binary tree with $\sqrt{n}$ leaves. The $i^{th}$ node on every path is connected to the $i^{th}$ leaf in the tree. There are $\sqrt{n}$ sessions with $s_i, t_i$ being the first and last node on the $i^{th}$ path. All capacities and demands are one. The graph's diameter is $\Theta(\log n)$, but its coding makespan is $\tilde{\Omega}(\sqrt{n})$.}
		\label{fig:dist-graph}
		\vspace{-0.6cm}
	\end{center}	
\end{figure}

A more sophisticated argument is needed to certify the $\tilde{\Omega}(\sqrt n)$ lower bound for this specific instance. The aforementioned papers~\cite{elkin2006unconditional,dassarma2012distributed,peleg2000near} prove their results by implicitly using the technique we formalize as our moving cut in the following definition and lemma (proven in \Cref{sec:moving-cut}).

\begin{restatable}{Def}{movingcut}(Moving cut)
   \label{def:moving-cut}
   Let $G = (V, E)$ be a communication network with capacities $c:E\rightarrow \mathbb{Z}_{\geq 1}$ and let $\{(s_i,t_i)\mid i\in [k]\}$ be source-sink pairs. A \textbf{moving cut} is an assignment $\ell : E \to \mathbb{Z}_{\ge 1}$ of positive integer lengths to edges of $G$. We say the moving cut has \textbf{capacity} $C$, if $\sum_{e \in E} c_e (\ell_e - 1) = C$, and \textbf{distance} $T$, if all sinks and sources are at distance at least $T$ with respect to $\ell$; i.e., $\forall i, j \in [k]$ we have $d_\ell(s_i, t_j) \geq T$.
\end{restatable}


\begin{restatable}{lem}{cutImpliesLowerBound}\label{cut-implies-lower-bound}
  Let $\calM = (G, \calS)$ be a unicast instance which admits a moving cut $\ell$ of capacity strictly less than $\sum_{i\in [k]} d_i$ and distance $T$. Then \emph{any} (coding) protocol for $\calM$ has makespan at least $T$.
\end{restatable}	

\Cref{cut-implies-lower-bound} can be seen as a natural generalization of \Cref{obs:cut-implies-coding-lb}, which can be equivalently restated in the following way: \textit{``Suppose that after increasing each edge $e$'s length from one to $\ell_e \in \{\pmb{1}, \pmb{T+1}\}$, we have that (1) $\sum_{e \in E} c_e (\ell_e - 1) < \sum_{i=1}^k d_i$, and (2) $\textrm{dist}_{\ell}(s_i, t_j) \ge T$. Then any (coding) protocol $\calM$ has makespan at least $T$''}. Dropping the restriction on $\ell_e$ recovers \Cref{cut-implies-lower-bound}.

Strikingly, the moving cut technique allows us not only to prove tight bounds (up to polylog factors) for the instance of \Cref{fig:dist-graph}---it allows us to get such tight bounds \emph{for every multiple-unicast instance}. 
In order to upper bound the makespan coding gap, we therefore relate such a moving cut with the optimal routing makespan, as follows.

To characterize the optimal routing makespan, we study hop-bounded multicommodity flow, which is an LP relaxation of routing protocols of makespan $T$. First, we show that a fractional LP solution of high value to this LP implies a routing with makespan $O(T)$.
Conversely, if the optimal value of this LP is low, then by strong LP duality this LP's dual has a low-valued solution, which we us to derive a moving cut and lower bound the coding makespan.
Unfortunately, the dual LP only gives us bounds on (average) distance between source-sink pairs $(s_i,t_i)$, and not between all sources $s_i$ and sinks $t_j$ (including $j\neq i$), as needed for moving cuts. For this conversion to work, we prove a generalization of the main theorem of Arora, Rao and Vazirani \cite{arora2009expander} to general metrics, of possible independent interest. (See \Cref{sec:all-to-all-pairs}.) 
This allows us to show that a low-valued dual solution implies a moving cut certifying that no coding protocol has makespan less than $T/O(\log k\cdot \log (\sum_i d_i /\min_i d_i))$. 
As the above rules out low-valued optimal solutions to the LP for $T=T^*\cdot O(\log k\cdot \log (\sum_i d_i /\min_i d_i))$ with $T^*$ the optimal coding makespan, the LP must have high optimal value, implying a routing protocol with makespan $O(T)$, and thus our claimed upper bound on the makespan coding gap.

\subsubsection{Lower Bounding the Coding Gap}
To complement our polylogarithmic upper bound on the makespan gap, we construct a family of multiple-unicast instances $\calM$ that exhibit a polylogarithmic makespan coding gap. We achieve this by amplifying the gap via graph products, a powerful technique that was also used in prior work to construct extremal throughput network coding examples~\cite{blasiak2011lexicographic,lovett2014linear,braverman2017coding}. 
Here we outline this approach, as well as the additional challenges faced when trying to use this approach for makespan. 

We use a graph product introduced by \citet{braverman2017coding} (with some crucial modifications).
\citet{braverman2017coding} use their graph product to prove a conditional \emph{throughput} coding gap similar to the one of \Cref{wc-gap-lb}, conditioned on the (unknown) existence of a multiple-unicast instance $\inst$ with non-trivial throughput coding gap. 
The graph product of \cite{braverman2017coding} takes instances $\inst_1, \inst_2$ and intuitively replaces each edge of $\inst_1$ with a source-sink pair of a different copy of $\inst_2$. More precisely, multiple copies of $\inst_1$ and $\inst_2$ are created and interconnected. Edges of a copy of $\inst_1$ are replaced by the same session of different copies of $\inst_2$; similarly, sessions of a copy of $\inst_2$ replace the same edge in different copies of $\inst_1$. This product allows for coding protocols in $\inst_1$ and $\inst_2$ to compose in a straightforward way to form a fast coding protocol in the product instance. The challenge is in proving impossibility results for routing protocols, which requires more care in the definition of the product graph. 

To address this challenge, copies of instances are interconnected along a high-girth bipartite graph to prevent unexpectedly short paths from forming after the interconnection. 
For example, to prove a throughput routing impossibility result, \citet{braverman2017coding} compute a dual of the multicommodity flow LP (analogous to our LP, but \emph{without any hop restriction}) to certify a limit on the routing performance. In the throughput setting, a direct tensoring of dual LP solutions of $\inst_1$ and $\inst_2$ gives a satisfactory dual solution of the product instance. In more detail, a dual LP solution in $\inst$ assigns a positive length $\ell_{\inst}(e)$ to each edge in $\inst$; each edge of the product instance corresponds to two edges $e_1 \in \inst_1$ and $e_2 \in \inst_2$, and the direct tensoring $\ell_+( (e_1, e_2) ) = \ell_{\inst_1}(e_1) \cdot \ell_{\inst_2}(e_2)$ provides a feasible dual solution with an adequate objective value. To avoid creating edges in the product distance of zero $\ell_+$-length, they contract edges assigned length zero in the dual LP of either instance. Unfortunately for us, such contraction is out of the question when studying makespan gaps, as such contractions would shorten the hop length of paths, possibly creating short paths with no analogues in the original instance.

Worse yet, any approach that uses the dual of our $T$-hop-bounded LP is bound to fail in the makespan setting. To see why, suppose we are given two instances $\inst_1, \inst_2$, both of which have routing makespan at least $T$ and expect that the product instance $\inst_+$ to have routing makespan at least $T^2$ by some construction of a feasible dual LP solution. Such a claim cannot be directly argued since a source-sink path in the product instance that traverses, say, $T-1$ different copies of $\inst_2$ along a path of hop-length $T+1$ in $\inst_2$ could carry an arbitrary large capacity! This is since the hop-bounded LP solution on $\inst_2$ only takes \emph{short} paths, of hop-length at most $T$, into account. Since there is no direct way to compose the dual LP solutions, we are forced to use a different style of analysis from the one of \cite{braverman2017coding}, which in turn forces our construction to become considerably more complicated.

To bound the routing \emph{makespan} in the product instance we rely on \Cref{obs:cut-implies-routing-lb}: We keep a list of edges $F$ along with each instance and ensure that (i) all source-sink distances in the $F$-deleted instance are large and that (ii) the ratio of the number of sessions $k$ to $|F|$ is large. We achieve property (i) by interconnecting along a high-girth graph and treating the replacements of edges in $F$ in a special way (hence deviating from the construction of \cite{braverman2017coding}). Property (ii) is ensured by making the inner instance $\inst_2$ significantly larger than the outer instance $\inst_1$, thus requiring many copies of $\inst_1$ and resulting in a large number of sessions in the product graph. To allow for this asymmetric graph product, we need an infinite number of base cases with non-trivial makespan coding gap for our recursive constructions (rather than a single base instance, as in the work of \citet{braverman2017coding}). This infinite family is fortunately obtained by appropriately generalizing the instance of \Cref{fig:gap-53}.

The main challenge in our approach becomes controlling the size of the product instance. To achieve this, we affix to each instance a relatively complicated set of parameters (e.g., coding makespan, number of edges, number of sessions, etc.) and study how these parameters change upon applying the graph product. Choosing the right set of parameters is key---they allow us to properly quantify the size escalation. In particular, we show that the coding gap grows doubly-exponentially and the size of the instances grow triply-exponentially, yielding the desired polylogarithmic coding gap.

\subsection{Related Work}\label{sec:related}
This work ties in to many widely-studied questions. We outline some of the most relevant here.

\paragraph{Routing multiple unicasts.} Minimizing the makespan of multiple unicasts using routing has been widely studied. When packets must be routed along fixed paths, two immediate lower bounds on the makespan emerge: \emph{dilation}, the maximum length of a path, and \emph{congestion}, the maximum number of paths crossing any single edge. A seminal result of \citet{leighton1994packet} proves one can route along such fixed paths in $O(\mathrm{congestion} + \mathrm{dilation} )$ rounds, making the result optimal up to constants. Follow ups include works 
improving the constants in the above bound \cite{peis2011universal,scheideler2006universal}, computing such protocols~\cite{leighton1999fast}, 
simplifying the original proof~\cite{rothvoss2013simpler}, routing in distributed models~\cite{rabani1996distributed, ostrovsky1997universal}, and so on. When one has the freedom to choose paths, \citet{bertsimas1999asymptotically} gave near-optimal routing solutions for asymptotically-large packet sizes, later extended to all packet sizes by \citet{srinivasan2001constant}. The power of routing for multiple unicasts is therefore by now well understood.

\paragraph{Network coding gains.} The utility of network coding became apparent after \citet{ahlswede2000network} proved it can increase the (single multicast) throughput of a communication network. Following their seminal work, there emerged a vast literature displaying the advantages of network coding over routing for various measures of efficiency in numerous communication models, including for example energy usage in wireless networks \cite{wu2005minimum,goel2008network,fragouli2008efficient}, delay minimization in repeated single unicast\cite{chekuri2015delay, wang2016sending}, and makespan in gossip protocols \cite{deb2006algebraic,haeupler2016analyzing,haeupler2015simple}.
The throughput of a single multicast (i.e., one single node sending to some set of nodes), arguably the simplest non-trivial communication task, was also studied in great detail (e.g., \cite{ahlswede2000network, li2004network2, agarwal2004advantage, ho2006random, li2009constant, huang2013space}). In particular, \citet{agarwal2004advantage} showed that the throughput coding gap for a single multicast equals the integrality gap of natural min-weight Steiner tree relaxations, for which non-trivial bounds were known (see, e.g., \cite{zosin2002directed, halperin2007integrality}). 
While the throughput coding gap for a single multicast is now fairly well understood, the case of multiple senders seems to be beyond the reach of current approaches.

\paragraph{Throughput gaps for multiple unicasts.} 
The routing throughput for multiple unicasts is captured by multicommodity max-flow, while the coding throughput is clearly upper bounded by the sparsest cut. Known multicommodity flow-cut gap bounds therefore imply the throughput coding gap for $k$ unicasts is at most $O(\log k)$ \cite{linial1995geometry,aumann1998log}, and less for special families of instances~\cite{rao1999small,lee2010genus,chekuri2013flow,krauthgamer2019flow,chekuri2006embedding,klein1993excluded,chekuri2005multicommodity}.
In 2004 \citet{li2004network} and \citet{harvey2004comparing} independently put forward the \emph{multiple-unicast conjecture}, which asserts that the throughput coding gap is trivial (i.e., it is one). 
This conjecture was proven true for numerous classes of instances \cite{hu1963multi,okamura1981multicommodity,jain2006capacity,kramer2006edge,adler2006capacity}.
More interestingly, a positive resolution of this conjecture has been shown to imply unconditional lower bounds in external memory algorithm complexity~\cite{adler2006capacity, farhadi2019lower}, computation in the cell-probe model~\cite{adler2006capacity}, and (recently) an $\Omega(n \log n)$ circuit size lower bound for multiplication of $n$-bit integers~\cite{afshani2019lower} (matching an even more recent breakthrough algorithmic result for this fundamental problem \cite{harvey2019polynomial}).
Given this last implication, it is perhaps not surprising that despite attempts by many prominent researchers~\cite{jain2006capacity,harvey2006capacity,yin2018reduction,li2004network,xiahou2014geometric,langberg2009multiple,al2008k}, the conjecture remains open and has established itself as a notoriously hard open problem. Indeed, even improving the $O(\log k)$ upper bound on throughput coding gaps seems challenging, and would imply unconditional \emph{super-linear circuit size lower bounds}, by the work of \citet{afshani2019lower}.
Improving our upper bound on makespan coding gaps to $o(\log k)$ would directly imply a similar improvement for throughput coding gaps, together with these far-reaching implications.



\section{Upper Bounding the Coding Gap}\label{sec:upper-bounds}

In this section we prove \Cref{wc-gap-ub}, upper bounding the makespan network coding gap. Given a multiple-unicast instance $\calM$ we thus want to upper bound its routing makespan and lower bound its coding makespan. To characterize these quantities we start with a natural hop-bounded multicommodity flow LP, \primal$(T)$, which serves as a ``relaxation'' of routing protocols of makespan at most $T$. The LP, given in \Cref{fig:primaldual-flows}, requires sending a flow of magnitude $z\cdot d_i$ between each source-sink pair $(s_i,t_i)$, with the additional constraints that (1) the combined congestion of any edge $e$ is at most $T \cdot c_e$ where $c_e$ is the capacity of the edge (as only $c_e$ packets can use this edge during any of the $T$ time steps of a routing protocol), and (2) the flow is composed of only \emph{short} paths, of at most $T$ hops. Specifically, for each $i\in [k]$, we only route flow from $s_i$ to $t_i$ along paths in $\calP_i(T)\defeq \{p:s_i\pathto t_i \mid |p|\leq T,\, p \text{ is simple}\}$, the set of simple paths of hop-length at most $T$ connecting $s_i$ and $t_i$ in $G$.

\begin{figure}[h]
	\begin{center}
		\begin{tabular}{rl|rl}
			\multicolumn{2}{c|}{Primal: \primal$(T)$} & \multicolumn{2}{c}{Dual:
				\dual$(T)$} \\ \hline  
			maximize & $z$ & minimize & $T \cdot \sum_{e \in E} c_e \ell_e$ \\
			subject to: & & subject to: & \\
			
			$\forall i\in [k]$: &  $\sum_{p\in \calP_i(T)} f_i(p) \geq z \cdot d_i$ & $\forall i\in [k], p\in \calP_i(T)$: & $\sum_{e\in p} \ell_e \geq h_i$ \\
			
			$\forall e \in E$: & $\sum_{p\ni e} f_i(p) \leq T \cdot c_e$ & & $\sum_{i \in [k]} d_i h_i \geq 1$ \\
			
			$\forall i\in [k],p$: & $f_i(p) \geq 0$ & $\forall e\in E$: & $\ell_e \geq0$ \\
			
			& & $\forall i\in [k]$: & $h_i\geq 0$				
		\end{tabular}
	\end{center}
	\vspace{-0.5cm}
	\caption{The concurrent flow LP relaxation and its dual.}
	\label{fig:primaldual-flows}
\end{figure}

A routing protocol solving $\calM$ in $T$ rounds yields a solution to \primal($T$) of value $z = 1$, almost by definition.\footnote{Such a protocol must send $d_i$ packets along paths of length at most $T$ between each sour-sink pair $(s_i,t_i)$, and it can send at most $c_e$ packets through each edge $e$ during any of the $T$ rounds, or at most $c_e\cdot T$ packets overall.}
A partial converse is also true; a feasible solution to \primal($T$) of value at least $\Omega(1)$ implies a routing protocol for $\calM$ in time $O(T)$. 
This can be proven using standard LP rounding \cite{srinivasan2001constant} and $O(\textrm{congestion + dilation})$ path routing \cite{leighton1994packet}. (See \Cref{sec:upper-bounds-appendix}.)

\begin{restatable}{prop}{primalroutingprotocol}\label{lem:primal-to-routing-protocol}
	Let $z, \{f_i(p) \mid i\in [k], p\in \calP_i(T)\}$ be a feasible solution for \primal$(T)$. Then there exists an integral routing protocol with makespan $O(T/z)$. 
\end{restatable}

Complementing the above, we show that a low optimal LP value for \primal$(T)$ implies that no \emph{coding} protocol can solve the instance in much less than $T$ time.

\begin{restatable}{lem}{dualtohard}\label{lem:dual-to-coding-hardness}
	If the optimal value of \primal$(T)$ is at most $z^*\leq 1/10$, then the coding makespan for $\calM$ is at least $T/(C\cdot \log(k)\cdot \log\left(\sum_i d_i / \min_i d_i\right) )$ for some constant $C>0$.
\end{restatable}

Before outlining our approach for proving \Cref{lem:dual-to-coding-hardness}, we show why this lemma together with \Cref{lem:primal-to-routing-protocol} implies our claimed upper bound for the makespan network coding gap.
\latencyub*
\begin{proof}
	Fix a multiple-unicast instance $\calM$. Let $T^*$ be the minimum makespan for any coding protocol for $\calM$. Let $T=(C+1)\cdot T^*\cdot (\log(k)\cdot \log(\sum_i d_i /\min_i d_i))$ for $C$ as in \Cref{lem:dual-to-coding-hardness}. Then, the LP \primal$(T)$ must have optimal value at least $z^*\geq 1/10$, else by \Cref{lem:dual-to-coding-hardness} and our choice of $T$, any coding protocol has makespan at least $T/O(\log (k)\cdot \log(\sum_i d_i/\min_i d_i)) > T^*$, contradicting the definition of $T^*$. But then, by \Cref{lem:primal-to-routing-protocol}, there exists a routing protocol with makespan $O(T/z^*) = O(T^* \cdot \log(k)\cdot \log(\sum_i d_i/\min_i d_i)$. The theorem follows.
\end{proof}

The remainder of the section is dedicated to proving \Cref{lem:dual-to-coding-hardness}. That is, proving that a low optimal value for the LP \primal$(T)$ implies a lower bound on the makespan of any coding protocol for $\calM$.
To this end, we take a low-valued LP solution to the dual LP \dual$(T)$ (implied by strong LP duality) and use it to obtain an information-theoretic certificate of impossibility, which we refer to as a \emph{moving cut}. \Cref{sec:moving-cut} introduces a framework to prove such certificates of impossibility, which we show \emph{completely characterizes} any instance's makespan (up to polylog terms). We then explain how to transform a low-value dual LP solution to such a moving cut in \Cref{sec:dual-to-cut}. For this transformation, we prove a lemma reminiscent of \citet[Theorem 1]{arora2009expander} for general metrics, in \Cref{sec:all-to-all-pairs}.

\subsection{Moving Cuts: Characterizing Makespan}\label{sec:moving-cut}
In this section we prove that moving cuts characterize the makespan of a multiple unicast instance. For ease of reference, we re-state the definition of moving cuts.

\movingcut*

We start by proving \Cref{cut-implies-lower-bound}, whereby moving cuts with small capacity and large distance imply makespan lower bounds. 

%

\cutImpliesLowerBound*
\begin{proof}
	We will show via simulation that a protocol solving $\calM$ in at most $T-1$ rounds would be able to compress $\sum_{i=1}^k d_i$ random bits to a strictly smaller number of bits, thereby leading to a contradiction. Our simulation proceeds as follows.
	We have two players, Alice and Bob, who control different subsets of nodes. 
	In particular, if we denote by $A_r \defeq \{v\in V \mid \min_i \textrm{dist}_\ell(s_i,v)\leq r\}$ the set of nodes at distance at most $r$ from any source, then during any round $r \in \{0, 1, \ldots, T-1\}$ all nodes in $A_r$ are ``spectated'' by Alice. 
	By spectated we mean that Alice gets to see all of these nodes' private inputs and received transmissions during the first $r$ rounds. 
	Similarly, Bob, at time $r$, spectates $B_r \defeq V \setminus A_r$. 
	Consequently, if at round $r$ a node $u \in V$ spectated by Alice sends a packet to a node $v \in V$, then 
	Bob will see the contents of that packet if and only if Bob spectates the node $v$ at round $r+1$. 
	That is, this happens only if $u\in A_r$ and $v\in B_{r+1} = V\setminus A_{r+1}$. Put otherwise, Bob can receive a packet from Alice along edge $e$ during times $r\in [\min_i \textrm{dist}_\ell(s_i,u),\min_i \textrm{dist}_\ell(s_i,v)-1]$. Therefore, the number of rounds transfer can happen along edge $e$ is at most $\min_i \textrm{dist}_\ell(s_i, v) - \min_i \textrm{dist}_\ell(s_i, u) - 1 \leq  \ell_e - 1$. Hence, the maximum number of bits transferred from Alice to Bob via $e$ is $c_e(\ell_e - 1)$. Summing up over all edges, we see that the maximum number of bits Bob can ever receive during the simulation is at most $\sum_{e \in E} c_e (\ell_e - 1) < \sum_{i=1}^k d_i$. 
	Now, suppose Alice has some $\sum_{i=1}^k d_i$ random bits.
	By simulating this protocol with each source $s_i$ having (a different) $d_i$ of these bits, we find that if all sinks receive their packet in $T$ rounds, then Bob (who spectates all $t_j$ at time $T-1$, as $\min_i \textrm{dist}_\ell(s_i,t_j) \geq T$ for all $j$) learns all $\sum_{i=1}^k d_i$ random bits while receiving less than $\sum_{i=1}^k d_i$ bits from Alice---a contradiction.
\end{proof}

\Cref{cut-implies-lower-bound} suggests the following recipe for proving makespan lower bounds:
Prove a lower bound on the makespan of some sub-instance ${\calM}' = (G, {\calS}')$ with $\calS' \subseteq \calS$ induced by indices $I\subseteq [k]$ using \Cref{cut-implies-lower-bound}. As any protocol solving $\calM$ solves $\calM'$, a lower bound on the makespan of $\calM'$ implies a lower bound on the makespan of $\calM$. 
So, to prove makespan lower bounds for $\calM$, identify a moving cut for some sub-instance of $\calM$. If this moving cut has capacity less than the sum of demands of the sub-instance and distance at least $T$, then the entire instance, $\calM$, has makespan at least $T$.

By the above discussion, the worst distance of a (low capacity) moving cut over any sub-instance serves as a lower bound on the makespan of any instance.
The following lemma, whose proof is deferred to the end of \Cref{sec:dual-to-cut}, asserts that in fact, the highest distance of a moving cut over any sub-instance is equal (up to polylog terms) to the best routing makespan of $\calM$. Consequently, by \Cref{wc-gap-ub}, the strongest lower bound obtained using moving cuts is equal up to polylog terms to the optimal (coding) makespan for $\calM$. 

\begin{restatable}{lem}{onlyObstruction}\label{lem:only-obstruction}
	If a $k$-unicast instance $\mathcal{M}$ has no routing protocol with makespan $T$, then
	there exists a set of sessions $I \subseteq [k]$ with a moving cut of capacity strictly less than $\sum_{i\in I}d_i$ and distance at least $T / O(\log k\cdot (\sum_{i} d_i/\min_i d_i))$ with respect to the unicast sub-instance induced by $I$.
\end{restatable}

We now turn to leveraging moving cuts to prove makespan lower bounds. Specifically, \Cref{lem:dual-to-coding-hardness}.

\subsection{From Dual Solution to Moving Cut}\label{sec:dual-to-cut}

\Cref{lem:primal-to-routing-protocol} shows that high objective value for the primal LP, \primal($T$) implies an upper bound on the routing time for the given instance.
In this section we prove the ``converse'', \Cref{lem:dual-to-coding-hardness}, whereby low objective value of the primal LP implies a lower bound on the \emph{coding} time for the given instance. 

Our approach will be to prove that a low objective value of the primal LP
---implying a feasible dual LP solution of low value---yields a moving cut for some sub-instance. This, by \Cref{cut-implies-lower-bound}, implies a lower bound on protocols for this sub-instance, and thus for the entire instance, from which we obtain \Cref{lem:dual-to-coding-hardness}. We turn to converting a low-valued dual LP solution to such a desired moving cut.

By definition, a low-value feasible solution to the dual LP, \dual$(T)$, assigns non-negative lengths $\ell : E \to \mathbb{R}_{\ge 0}$ such that (1) the $c$-weighted sum of $\ell$-lengths is small, i.e., $\sum_{e \in E} c_e \ell_e = \tO(1 / T)$, as well as (2) if $h_i$ is the $\ell$-length of the $T$-hop-bounded $\ell$-shortest path between $s_i$ and $t_i$, then $\sum_{i \in [k]} d_i\cdot h_i \ge 1$. Property (1) implies that appropriately scaling the lengths $\ell$ yields a moving cut given by lengths $\tilde{\ell}$ of bounded capacity, $\sum_e c_e \tilde{\ell}_e$. 
For this moving cut to be effective to lower bound the makespan of some sub-instance using \Cref{cut-implies-lower-bound}, the cut must have high distance w.r.t.~this sub-instance. As a first step to this end, we use Property (2) to identify a subset of source-sink pairs $I\subseteq [k]$ with pairwise $\tilde{\ell}$-distance at least $\tilde{\Omega}(T)$. 
\Cref{lem:bucketing-lemma}, proven in \Cref{sec:upper-bounds-appendix} using a ``continuous'' bucketing argument, does just this. The claim introduces a loss factor $\agap$ that we use throughout this section.
\begin{restatable}{cla}{bucketing}\label{lem:bucketing-lemma}
	Given sequences $h_1, \ldots, h_k, d_1, \ldots, d_k \in \reals_{\ge 0}$ with $\sum_{i\in [k]} d_i \cdot h_i \ge 1$ there exists a non-empty subset $I \subseteq [k]$ with $\min_{i \in I} h_i \ge \frac{1}{\agap \cdot \sum_{i\in I}d_i}$ for $\agap \in \left[1,O\left( \log \frac{\sum_{i\in [k]} d_i}{\min_{i\in [k]} d_i} \right )\right]$.
\end{restatable}

Scaling up the $\ell$ lengths yields a sub-instance induced by pairs $I\subseteq [k]$ and moving cut with bounded capacity and with $\tilde{\ell}$-distance between every source and its sink of at least $\tilde{\Omega}(T)$, i.e., $d_{\tilde{\ell}}(s_i,t_i)\geq \tilde{\Omega}(T)$ for all $i\in I$. However, \Cref{cut-implies-lower-bound} requires $\tilde{\ell}$-distance $\tilde{\Omega}(T)$ between \emph{any} source and sink, i.e., $d_{\tilde{\ell}}(s_i,t_j) \geq \tilde{\Omega}(T)$ for all $i,j\in I$.
To find a subset of source-sink pairs with such distance guarantees, we rely on the following metric decomposition lemma, whose proof is deferred to \Cref{sec:all-to-all-pairs}.

\begin{restatable}{lem}{lemMetricAllToAll}\label{lemMetricAllToAll}
  Let $(X, d)$ be a metric space. Given $n$ pairs $\{ (s_i, t_i) \}_{i \in [n]}$ of points in $X$ with at most $k$ distinct points in $\bigcup_i\{s_i,t_i\}$ and pairwise distances at least $d(s_i, t_i) \geq T$, there exists a subset of indices $I \subseteq [n]$ of size $|I|\geq \frac{n}{9}$ such that $d(s_i, t_j) \geq \frac{T}{O(\log k)}$ for all $i,j\in I$. Moreover, such a set can be computed in polynomial time.
\end{restatable}

We are now ready to construct the moving cut.

\begin{lem}\label{lem:small-dual-to-moving-cut}
 If the optimal value of \primal$(T)$ is at most $z^*\leq 1/10$, then there exists $\tilde{I} \subseteq [k]$ and a moving cut $\tl$ of capacity strictly less than $\sum_{i\in \tI}d_i$ and distance at least $T / O(\agap \log k)$ with respect to the unicast sub-instance induced by $\tilde{I}$ (i.e., $\tilde{\calM} = (G, \tilde{\calS})$ where $\tilde{\calS} = \{(s_i, t_i, d_i)\}_{i \in \tI}$).
\end{lem}
\begin{proof}
	  By strong duality, the dual LP \dual$(T)$  has a feasible solution $\{h_i, \ell_e \mid i\in [k], e\in E\}$ to \dual$(T)$ with objective value $T \sum_{e \in E} c_e \ell_e = z^*$. Fix such a solution.
	  Let $I\subseteq [k]$ be a subset of indices as guaranteed by  \Cref{lem:bucketing-lemma}. Define $\tl_e \defeq 1 + \lfloor \ell_e \cdot T \cdot \sum_{i\in I} d_i \rfloor$ for all $e \in E$ and note that $\tl_e \in \integers_{\ge 1}$. 
  By definition of $\tl$ and $T\sum_e c_e\ell_e = z$, we get a bound on the capacity of $\tl$:
  \begin{align*}
    \sum_{e \in E} c_e (\tl_e - 1) & \le \sum_e c_e\ell_e \cdot T \cdot \sum_{i\in I}d_i = z^* \cdot \sum_{i\in I} d_i \leq \frac{1}{10} \sum_{i \in I} d_i < \frac{1}{9} \sum_{i \in I} d_i.
  \end{align*}
 
  We now show that $d_{\tl}(s_i, t_i) > T/\agap$ for all $i \in I$. Consider any simple path $p$ between $s_i \pathto t_i$. Denote by $\tl(p)$ and $\ell(p)$ the length with respect to $\tl$ and $\ell$, respectively. It is sufficient to show that $\tl(p) > T/\agap$. If $p \not \in \calP_i(T)$, i.e., the hop-length of $p$ (denoted by $|p|$) is more than $T$. Then $\tl(p) \ge |p| > T \geq T/\agap$, since $\tl_e \ge 1\ \forall e \in E$. Conversely, if $p \in \calP_i(T)$, then by our choice of $I$ as in \Cref{lem:bucketing-lemma} and the definition of $h_i$, we have that $\ell(p) \ge h_i \ge \frac{1}{\agap \sum_{i\in I} d_i}$, hence 
   \begin{align*}
     \tl(p) & \geq \ell(p) \cdot T\cdot  \sum_{i\in I} d_i = \frac{1}{\agap \sum_{i\in I} d_i} \cdot T \cdot \sum_{i\in I} d_i = T/\agap.
   \end{align*}

   Finally, we choose a subset $\tI \subseteq I$ s.t. $d_{\tl}(s_i, t_j) > T / O(\agap \log k)$ for all $i, j\in \tI$. By \Cref{lemMetricAllToAll} applied to the graphic metric defined by $\tl$ and each pair $(s_i,t_i)$ repeated $d_i$ times, there exists a multiset of indices $\tI\subseteq I \subseteq [k]$ such that $d_{\tl}(s_i,t_j)\geq T / O(\agap \log k)$ for all $i,j\in \tI$ and such that $|\tI| \geq \sum_i d_i/9$. Therefore, taking each pair $(s_i,t_i)$ indexed by $\tI$ at least once, we find a subset of sessions $\tI \subseteq [k]$ such that $\sum_{i\in \tI} d_i \geq \frac{1}{9} \sum_{i\in [k]} d_i > \sum_e c_e\cdot (\tl_e-1)$ and $d_{\tl}(s_i,t_j)\geq T/O(\agap \log k)$ for all $i,j\in \tI$. In other words, $\tl$ is a moving cut of capacity strictly less than $\sum_{i\in \tI}d_i$ and distance $T / O(\agap \log k)$ with respect to the sub-instance induced by $\tI$.
\end{proof}

Combining \Cref{lem:small-dual-to-moving-cut} with \Cref{cut-implies-lower-bound}, we obtain this section's main result, \Cref{lem:dual-to-coding-hardness}.

\dualtohard*

\paragraph{Remark 1.} We note that the $\log k$ term in \Cref{lem:dual-to-coding-hardness}'s bound is due to the $\log k$ term in the bound of \Cref{lemMetricAllToAll}. For many topologies, including genus-bounded and minor-free networks, this $\log k$ term can be replaced by a constant (see \Cref{sec:all-to-all-pairs}), implying smaller makespan gaps for such networks.

\paragraph{Remark 2.} \Cref{lem:only-obstruction}, which states that a lower bound on routing makespan implies the existence of a moving cut of high distance with respect to some sub-instance, follows by \Cref{lem:small-dual-to-moving-cut} and \Cref{lem:primal-to-routing-protocol}, as follows. By \Cref{lem:primal-to-routing-protocol}, $\mathcal{M}$ having no routing protocol with makespan $T$ implies that for some constant $c>0$, the LP \primal($c\cdot T$) has objective value at most $1/10$. \Cref{lem:small-dual-to-moving-cut} then implies the existence of the moving cut claimed by \Cref{lem:only-obstruction}.

\subsection{From Pairwise to All-Pairs Distances}\label{sec:all-to-all-pairs}

This section is dedicated to a discussion and proof of the following Lemma that seems potentially useful beyond the scope of this paper.

\lemMetricAllToAll* 

We note that the above lemma is similar to the main Theorem of \citet{arora2009expander}. Our result holds for general metrics with a factor of $O(\log k)$ in the distance loss, while their holds for $\ell_2^2$ metrics with a factor of $O(\sqrt{\log k})$. The results are incomparable and both are tight. (The tightness of 
\Cref{lemMetricAllToAll} can be shown to be tight for graph metrics, for example in graph metrics of constant-degree expanders.)

To prove \Cref{lemMetricAllToAll} we rely on \emph{padded decompositions}~\cite{gupta2003bounded}. 
To define these, we introduce some section-specific notation. Let $(X, \textrm{dist})$ be a metric space. Let the (weak) diameter of a set of points $U \subseteq X$ be denoted by $\diam(U) \defeq \max_{x, y \in U} \textrm{dist}(x, y)$. 
We say a partition $P=\{X_1,X_2,\dots,X_t\}$ of $X$ is \emph{$\Delta$-bounded} if $\diam(X_i)\leq \Delta$ for all $i$.
Next, for $U \subseteq X$ and a partition $P$ as above, we denote by $U \subseteq P$ the event that there exists a part $X_i \in P$ containing $U$ in its entirety; i.e., $U \subseteq X_i$. Let $B(x, \rho) \defeq  \{ y \in X \mid \textrm{dist}(x, y) \le \rho \}$ denote the ball of radius $\rho \ge 0$ around $x \in X$.

\begin{Def}\label{def:padded}
	Let $(X, \textrm{dist})$ be a metric space. We say that a distribution $\calP$ over $\Delta$-bounded  partitions of $X$ is  $(\beta,\Delta)$-\emph{padded} if, for some universal constant $\delta$, it holds that for every $x \in X$ and $0\leq \gamma \leq \delta$,
	\begin{align*}
		\Pr_{P\sim \calP}[B(x, \gamma\Delta) \not \subseteq P] \le \beta \gamma.
	\end{align*}
\end{Def}
In words, each part of the partition has diameter at most $\Delta$ and the probability of any point $x$ in the metric being at distance less than $\gamma\Delta$ from a different part than its own part is at most $\beta\gamma$. 
Such decompositions were presented, for example, by \Citet{gupta2003bounded}.

\begin{lem}[\cite{gupta2003bounded}]
	Any metric $(X, \textrm{dist})$ on $k$ points admits a $(\beta,\Delta)$-padded decomposition, for any $\Delta>0$ and some $\beta=O(\log k)$. Such a decomposition can be computed in polynomial time.
\end{lem}

We are now ready to prove \Cref{lemMetricAllToAll}.

\begin{proof}[Proof of \Cref{lemMetricAllToAll}]
	First note that we can focus on the metric space induced by the $k$ distinct points. Let $\calP$ be a $\Delta$-bounded $\beta$-padded decomposition with $\Delta = T-1$ and $\beta=O(\log k)$. We first note that for all $i \in [k]$, $s_i$ and $t_i$ are contained in different parts since the diameter of each part $X_i$ is at most $\Delta = T-1$ and $\textrm{dist}(s_i, t_i) \geq T$. Furthermore, letting $\gamma=\frac{1}{2\beta}$, we have that 
	$\Pr[B(s_i, \gamma\Delta) \subseteq P] \ge \frac{1}{2}$. Let $I' \subseteq [n]$ be the subset of indices $i$ with $B(s_i, \gamma\Delta) \subseteq P$. Then we have $\Pr[i \in I'] \ge \frac{1}{2}$ for all $i \in [n]$.
	
	Flip a fair and independent coin for each part in $P$. Let $U \subseteq X$ be the set of points in parts whose coin came out heads, and $V \subseteq X$ be the analogous set for tails. Then for each $i \in I'$ we have that $\Pr[s_i \in U \text{ and } t_i \in V] = \frac{1}{4}$. Let $I \subseteq I'$ be the subset of indices $i$ with $s_i \in U \text{ and } t_i \in V$, giving $\Pr[i \in I] = \Pr[i \in I'] \cdot \Pr[i \in I \mid i \in I'] = \frac{1}{2} \cdot \frac{1}{4} = \frac{1}{8}\ \forall i \in [n]$. 
	We also have that $\textrm{dist}(s_i, t_j) > \rho = \frac{T-1}{2\beta}$ for all $i,j\in I$, since $B(s_i, \rho) \subseteq U$ for all $i \in I \subseteq I'$ and $\{ t_j \}_{j \in I} \cap U = \emptyset$. Therefore, this random process yields a subset of indices $I\subseteq[n]$ such that $\textrm{dist}(s_i,t_j) > \frac{T-1}{2\beta}$ for all $i,j\in I$, of expected size at least $\mathbb{E}[|I|] \geq \sum_{i\in [n]} \Pr[i\in I] \geq \frac{n}{8}$. As $n-\mathbb{E}[|I|]$ is a non-negative random variable, Markov's inequality implies that with constant probability $n-|I|\leq \frac{64}{63}\cdot(n-\mathbb{E}[|I|]) \leq \frac{8n}{9}$. 
	The lemma follows.
\end{proof}

\noindent\textbf{Remark:} The $O(\log k)$ term in \Cref{lemMetricAllToAll}'s bound is precisely the smallest possible $\beta$ for which $(\beta,\Delta)$-padded decompositions of the metric exist. For many graphic metrics, such as those of minor-excluded, bounded-genus, and bounded-doubling-dimension networks, padded decompositions with smaller $\beta$ exist \cite{lee2010genus,gupta2003bounded,abraham2019cops}. This improves the bounds of \Cref{lemMetricAllToAll} and thus \Cref{lem:dual-to-coding-hardness} by $(\log k)/\beta$, implying the same improvement for our makespan coding gaps for such networks.


\section{Polylogarithmic Coding Gap Instances}\label{sec:lower-bounds}

In this section we construct a family of multiple-unicast instances with polylogarithmic makespan coding gap. More precisely, we construct instances where the coding gap is at least $(5/3)^{2^i}$ and the size (both the number of edges and sessions) is bounded by $2^{2^{O(2^i)}}$. Here we give a bird's eye view of the construction and leave the details to subsequent subsections. We clarify that all big-O bounds like $f = O(g)$ mean there exists a universal constant $c > 0$ s.t. $f \le c \cdot g$ for all admissible values (in particular, there is no assumption on $f$ or $g$ being large enough).

We use the graph product of \cite{braverman2017coding} as our main tool. Given two multiple-unicast instances $\inst_1, \inst_2$ (called the \emph{outer} and \emph{inner instance}, resp.) we create a new instance $\inst_+$ where the coding gap is the product of the coding gaps of $\inst_1$ and $\inst_2$. The product is guided by a \emph{colored bipartite graph} $B = (V_1, V_2, E)$ where each edge is labeled by $(\chi_1, \chi_2) =$ (edge in $\inst_1$, session in $\inst_2$). Precisely, we create $|V_1|$ copies of $\inst_1$, $|V_2|$ copies of $\inst_2$ and for each edge $(a, b) \in E(B)$ with label $(\chi_1, \chi_2)$ we replace the edge $\chi_1$ in the $a^{th}$ copy of $\inst_1$ with session $\chi_2$ in the $b^{th}$ copy of $\inst_2$.

To prove a lower bound on the coding gap, one needs to upper bound the coding makespan and lower bound the coding makespan. The former is easy: the coding protocols nicely compose. The latter, however, is more involved. Our main tool is \Cref{obs:cut-implies-routing-lb}, which necessitates (i) keeping track of \emph{cut edges} $F$ along each instance $\inst$ such that all source-sink pairs of $\inst$ are well-separated after edges in $F$ are deleted, and (ii) keeping the ratio $r \defeq \frac{k}{|F|}$, number of sessions to cut edges, high. We must ensure that the properties are conserved in the product instance $\inst_+$. For (i), i.e., to disallow any short paths from forming as an unexpected consequence of the graph product, we choose $B$ to have \emph{high girth}. Also, we replace edges $F$ in the outer instances with paths rather than connecting them to a session in the inner instance. Issue (ii) is somewhat more algebraically involved but boils down to ensuring that the ratio of sessions to cut edges (i.e., $r$) in the inner instance is comparable to the size (i.e., number of edges) of the outer instance itself. Note that makes the size of the outer instance $\inst_1$ insignificant when compared to the size of the inner instance $\inst_2$.

We recursively define a family of instances by parametrizing them with a ``level'' $i \ge 0$ and a lower bound on the aforementioned ratio $r$, denoting them by $I(i, r)$. We start for $i=0$ with the $\nicefrac{5}{3}$ instance of \Cref{fig:gap-53} where we can control the ratio the aforementioned ratio $r$ by changing the number of sessions $k$ (at the expense of increasing the size). Subsequently, an instance on level $i$ is defined as a product two of level $i-1$ instances with appropriately chosen parameters $r$. One can show that the coding makespan for a level $i$ instance is at most $5^{2^i}$ and routing makespan is at least $3^{2^i}$, hence giving a coding gap of $(5/3)^{2^i}$. Furthermore, we show that the size of $I(i,r)$ is upper bounded by $r^{2^{O(2^i)}}$, giving us the full result.

Finally, we note an important optimization to our construction and specify in more detail how $I(i, r)$ is defined. Specifically, it is defined as the product of $\inst_1 \defeq I(i-1, 3r)$ being the outer instance and $\inst_2 \defeq I(i-1, m_1 / f_1)$ being the inner instance, where $m_1$ and $f_1$ are the number of total and cut edges of $\inst_1$. This necessitates the introduction and tracking of another parameter $u \defeq m / f$ to guide the construction. We remark that this might be necessary since if one uses a looser construction of $\inst_2 \defeq I(i-1, m_1)$ the end result $I(i, r)$ would be of size $r^{2^{O(i \cdot 2^i)}}$ and give a coding gap of $\exp\left(\frac{\log \log k}{\log \log \log k}\right)$, just shy of a polylogarithm.

\subsection{Gap Instances and Their Parameters}
In this section we formally define the set of instance parameters we will track when combining the instances.

A \textbf{gap instance} $\inst = (G, \calS, F)$ is a multiple-unicast instance $\calM = (G, \calS)$ over a connected graph $G$, along with an associated set of \textbf{cut edges} $F \subseteq E(G)$. We only consider gap instances where the set of terminals is disjoint, i.e., $s_i \neq s_j, s_i \neq t_j, t_i \neq t_j$ for all $i \neq j$. Furthermore, edge capacities and demands are one; i.e., $c_e = 1\ \forall e \in E(G)$, and $d_i = 1\ \forall (s_i, t_i, d_i) \in \calS$.
A gap instance $\inst = (G, \calS, F)$ has \textbf{parameters} $(a, b, f, k, m, r, u)$ when:
\begin{itemize}
	\item $\calM$ admits a network coding protocol with makespan at most $a$.
	\item Let $\textrm{dist}_{G \setminus F}(\cdot, \cdot)$ be the hop-distance in $G$ after removing all the cut edges $F$. Then for all terminals $i \in [k]$ we have that $\textrm{dist}_{G \setminus F}(s_i, t_i) \ge b$.
	\item The number of cut edges is $f = |F|$.
	\item The number of sessions is $k = |\calS|$.
	\item The graph $G$ has at most $m$ edges; i.e., $|E(G)| \le m$.
	\item $r$ is a lower bound on the ratio between number of sessions and cut edges; i.e., $k/f \ge r$.
	\item $u$ is an upper bound on the ratio between number of total edges and cut edges; i.e., $m/f \le u$.    
\end{itemize}

We note that the parameters of a gap instance immediately imply a lower bound on the optimal routing makespan via \Cref{obs:cut-implies-routing-lb}. Indeed, all packets transmitted in the first $b-1$ rounds must pass through $F$, and thus at most $f\cdot (b-1)$ packets can be sent between any source and its sink in the first $b-1$ rounds, implying that under any routing protocol, most sessions have completion time at least $b$.

\begin{obs}\label{lem:parameters-to-routing-lb}
  Let $\inst$ be a gap instance with parameters $(a, b, f, k, m, r, u)$ and $b \le r$. Then the routing makespan for $(G, \calS)$ is at least $b$. Moreover, for any routing protocol of $\inst$, at least $k \cdot (1 - \frac{b-1}{r})$ sessions have completion time at least $b$.
\end{obs}

As an application of the above observation, we obtain another proof of the lower bound of the routing makespan for the family of instances of \Cref{fig:gap-53}. More generally, letting the cut edges be the singleton $F = \{(S,T)\}$, we obtain a family of gap instances with the following parameters.

\begin{fact}\label{base-instance}
	The family of gap instances of \Cref{fig:gap-53} have parameters $(3, 5, 1, k, \theta(k^2), k, \theta(k^2))$ for $k \ge 5$.
\end{fact}

The above family of gap instances will serve as our base gap instances in a recursive construction which we describe in the following section.

\subsection{Graph Product of Two Gap Instances}

In this section we present the graph product that combines two instances to obtain one a with higher coding gap.

\begin{Def}
  \textbf{Colored bipartite graphs} are families of bipartite graphs $\calB({n_1,n_2,m,k,g})$. Graphs $B = (V_1,V_2,E) \in \calB({n_1,n_2,m,k,g})$ are bipartite graphs with $|V_1| = n_1$ (resp. $|V_2| = n_2$) nodes on the left (resp., right), each of degree $m$ (resp., $k$), and these graphs have girth at least $g$. In addition, edges of $B$ are colored using two color functions, \textbf{edge color} $\chi_1 : E(B) \rightarrow [m]$ and \textbf{session color} $\chi_2 : E(B) \rightarrow [k]$, which satisfy the following.
  \begin{itemize}
  \item \label{prop:replace-each-edge} $\forall v\in V_1$, the edge colors of incident edges form a complete set $\{\chi_1(e) \mid e\ni v\} = [m]$.
  \item \label{prop:use-each-session} $\forall v\in V_2$, the session colors of incident edges form a complete set $\{\chi_2(e) \mid e\ni v\} = [k]$.
  \item \label{prop:same-session-per-outer} $\forall v\in V_1$, the session colors of incident edges are unique $|\{\chi_2(e) \mid e\ni v\}| = 1$.
  \item \label{prop:same-edge-per-innter} $\forall v\in V_2$, the edge colors of incident edges are unique, i.e, $|\{\chi_1(e) \mid e\ni v\}| = 1$.
  \end{itemize}
\end{Def}

The size of the colored bipartite graphs will determine the size of the derived gap instance obtained by performing the product along a colored bipartite graph. The following gives a concrete bound on the size and, in turn, allows us to control the growth of the gap instances obtained this way.

\begin{lem}[\cite{braverman2017coding}]\label{colored-bip}
  $\forall r,m,g\geq 3$, there exists a colored bipartite graph $B \in \calB(n_1,n_2,m,k,2g)$ with $n_1,n_2\leq (9mk)^{g+3}$.
\end{lem}

\paragraph{Performing the product along a colored bipartite graph.}
Having defined colored bipartite graphs, we are now ready to define the graph product of $\inst_1$ and $\inst_2$ along $B$.

For $i \in \{1, 2\}$ let $\inst_i = (G_i, \calS_i, F_i)$ be a gap instance with parameters $(a_i, b_i, f_i, k_i, m_i, r_i, u_i)$ and let $B \in \calB(n_1, n_2, 2(m_1 - f_1), k_2, g)$ be a colored bipartite graph with girth $g \defeq 2 b_1 b_2$. We call $\inst_1$ the \textbf{outer instance} and $\inst_2$ the \textbf{inner instance}. Denote the \textbf{product} gap \textbf{instance} $\inst_+ \defeq T(\inst_1, \inst_2, B)$ by the following procedure:
\begin{itemize}
\item Replace each non-cut edge $e = \{u,v\} \in E(G_1) \setminus F_1$ with two anti-parallel arcs $\vec{e} = (u,v), \cev{e} = (v, u)$ and let $\vec{E}(G_1) = \{ \vec{e}_1, \vec{e}_2, \ldots, \vec{e}_{2(m_1 - f_1)} \}$ be the set of all such arcs.
\item Construct $n_1$ copies of $(V(G_1), \vec{E}(G_1))$ and $n_2$ copies of $G_2$. Label the $i^{th}$ copy as $G_1\at{i}$ and $G_2\at{i}$.
\item Every cut edge $e \in F_1$ and every $i \in [n_1]$ replace edge $e$ in $G_1\at{i}$ by a path of length $a_2$ with the same endpoints. Let $f_e\at{i}$ be an arbitrary edge on this replacement path.
\item For every $(i, j) \in E(B)$ where $i \in [n_1], j \in [n_2]$ with edge color $\chi_1$ and session color $\chi_2$ do the following. Let $\vec{e}\at{i}_{\chi_1} = (x, y)$ be the $\chi_1$\nth arc in $G_1\at{i}$ and let $(s, t)$ be the $\chi_2$\nth terminal pair in $G_2\at{j}$. Merge $x$ with $s$ and $y$ with $t$; delete $\vec{e}\at{i}_{\chi_1} = (x, y)$ from $G_1\at{i}$.
\item For each session in the outer instance $(s_i, t_i, d_i = 1) \in \calS_1$ add a new session $(s_i\at{j}, t_i\at{j}, 1)$ in $G_1\at{j}$ for $j \in [n_1]$ to the product instance.
\item The cut edges $F_+$ in the product instance $\inst_+$ consist of the the union of the following: (i) \emph{one} arbitrary (for concreteness, first one) edge from all of the $a_2$-length paths that replaced cut edges in $G_1\at{i}$, i.e., $\{ f_e\at{i} \mid e \in F_1, i \in [n_1] \}$, and (ii) all cut edges in copies of $G_2$, i.e., $\{ e\at{i} \mid e \in F_2, i \in [n_2] \}$.
\end{itemize}

We now give bounds on how the parameters change after combining two instances.
First, we note that by composing network coding protocols for $\inst_1$ and $\inst_2$ in the natural way yields a network coding protocol whose makespan is at most the product of these protocols' respective makespans.

\begin{restatable}{lem}{tensoredcodingtime}\label{lem:tensoring-coding}(Coding makespan)
  The product instance $\inst_+$ admits a network coding protocol with makespan at most $a_1 a_2$.
\end{restatable}

Less obviously, we show that if we choose a large enough girth $g$ for the colored bipartite graph, we have that the $b$ parameter of the obtained product graph is at least the product of the corresponding $b$ parameters of the inner and outer instances.

\begin{lem}[Routing makespan]\label{lem:tensoring-routing}
  Let $\inst_+ = (G_+, \calS_+, F_+)$ be the product instance using a colored bipartite graph $B$ of girth $g \defeq 2 b_1 b_2$ and let $\textrm{dist}_{G_+ \setminus F_+}(\cdot, \cdot)$ be the hop-distance in $G_+$ with all the edges of $F_+$ deleted. We have that $\textrm{dist}_{G_+ \setminus F_+}(s_i, t_i) \ge \min(b_1 b_2, \frac{g}{2}) = b_1 b_2$ for all $(s_i, t_i, d_i) \in \calS_+$.
\end{lem}
\begin{proof}
  Let $p$ be a path in $G_+ \setminus F_+$ between some terminals $s_i \pathto t_i$ that has the smallest hop-length among all $(s_i, t_i, d_i) \in \calS_+$. We want to show that $|p| \ge \min(b_1 b_2, \frac{g}{2})$.

  First, let $q$ be the path in the colored bipartite graph $B$ that corresponds to $p$. There are some technical issues with defining $q$ since merging vertices in the graph product has the consequence that some $v \in V(G_+)$ belong to multiple nodes $V(B)$. To formally specify $q$, we use the following equivalent rephrasing of the graph product that will generate an ``expanded instance'' $G_+'$. Instead of ``merging'' two vertices $u, v$ as in $G_+$, connect then with an edge $e$ of hop-length $h(e) = 0$ and add $e$ to $G_+'$. Edges from $G_+$ have hop-length $h(e) = 1$ and are analogously added to $G_+'$. The path $p$ can be equivalently specified as the path between $s_i \pathto t_i$ in $G_+' \setminus F_+$ that minimizes the distance $\textrm{dist}_h(s_i, t_i)$. Now, each vertex $V(G_+')$ belongs to exactly one vertex $V(B)$, hence the path $q$ in $B$ corresponding to $p$ is well-defined. Note that $p$ is a closed path in $G_+$ and $q$ is a closed path in $B$.

  Suppose that $q$ spans a non-degenerate cycle in $B$. Then $|p| \ge \frac{|q|}{2} \ge \frac{g}{2}$, where the last inequality $|q| \ge g$ is due to the girth of $G$. The first inequality $|p| \ge \frac{|q|}{2}$ is due to the fact that when $q$ enters a node $v \in V_2(B)$, a node representing an inner instance, the corresponding path $p$ had to traverse at least one inner instance edge before its exit since the set of terminals is disjoint and a path can enter/exit inner instances only in terminals.

  Suppose now that $q$ does not span a cycle in $B$, therefore the set of edges in $q$ span a tree $\calT$ in $B$ and $q$ is simply the (rotation of the unique) Eulerian cycle of that tree. Notation-wise, let $v \in V_1(B)$ be the node in the colored bipartite graph $B$ that contains the critical terminals $s_i$ and $t_i$ and suppose that $\calT$ is rooted in $v$. If the depth of $\calT$ is $1$ (i.e., consists only of $v \in V_1(B)$ and direct children $w_1, \ldots, w_t \in V_2(B)$), then $p$ must correspond to a $s_i \pathto t_i$ walk in $v$, where each (non-cut) edge traversal is achieved by a non-cut walk in the inner instance $w_j$ between a set of inner terminals. Note that every $s_i \pathto t_i$ non-cut walk has hop-length at least $b_1$ and each non-cut walk in the inner instance has hop-length at least $b_2$, for a cumulative $b_1 \cdot b_2$.

  Finally, suppose that $\calT$ has depth more than $1$, therefore there exists two $v, w \in V(\calT)$ and $v, w \in V_1(B)$. Since $\calT$ is traversed via an Eulerian cycle, the path $p$ passes through two terminals of $(s_j, t_j, \cdot) \in \calS_+$. Let $p'$ be the natural part of $p$ going from $s_j \pathto t_j$, e.g., obtained by clipping the path corresponding to the subtree of $q$ in $\calS$. Furthermore, let $p''$ be the part of the $p$ connecting $v$ and $w$ and is disjoint from $p''$. From the last paragraph we know that $|p''| \ge 1$ since it passes through at least one $u \in V_2(B)$. Also, by minimality of $s_i \pathto t_i$ we have that $\textrm{dist}_h(s_j, t_j) \ge \textrm{dist}_h(s_i, t_i)$. Now we have a contradiction since $\textrm{dist}_h(s_i, t_i) = |p| \ge |p'| + |p''| \ge 1 + \textrm{dist}_h(s_j, t_j)$.
\end{proof}

Combining Lemmas \ref{lem:tensoring-coding} and \ref{lem:tensoring-routing} together with some simple calculations (deferred to \Cref{sec:lower-bounds-appendix}), we find that the product instance is a gap instance with the following parameters.

\begin{restatable}{lem}{tensorparams}\label{lem:one-tensor-params}
  For $i \in \{1, 2\}$ let $\inst_i = (G_i, \calS_i, F_i)$ be a gap instance with parameters $(a_i, b_i, f_i, k_i, m_i, r_i, u_i)$ with $\frac{m_i}{f_i} \ge 2$ and $a_i \ge 2$; let $B \in \calB(n_1, n_2, 2(m_1 - f_1), k_2, 2 b_1 b_2)$ be a colored bipartite graph. Then $\inst_+ \defeq T(G_1, G_2, B)$ is a gap instance with parameters $a_+ \defeq a_1 a_2$, $b_+ \defeq b_1 b_2$, $f_+ \defeq n_1 f_1 + n_2 f_2$, $k_+ \defeq n_1 k_1$, $m_+ \defeq a_2 n_1 f_1 + n_2 m_2$, $r_+ \defeq r_1 \frac{1}{1 + 2u_1/r_2}$, $u_+ \defeq u_2 \frac{1 + a_2/2}{1 + r_2 / (2 u_1)}$. Moreover, $\frac{m_+}{f_+} \ge 2$ and $a_+ \ge 2$.
\end{restatable}

\subsection{Iterating the Graph Product}

Having bounded the parameters obtained by combining two gap instances, we are now ready to define a recursive family of gap instances from which we obtain our polylogarithmic makespan network coding gap.

\begin{Def}
  We recursively define a collection of gap instances $( \inst(i, r) )_{i \ge 0, r \ge 5}$, and denote its parameters by $(a_{i, r}, b_{i, r}, f_{i, r}, k_{i, r}, m_{i, r}, r_{i, r}, u_{i, r})$. For the base case, we let $\inst(0, r)$ be the gap instance of \Cref{base-instance} with parameters $(3, 5, 1, r, \theta(r^2), r, \theta(r^2))$. For $i+1 > 0$ we define $\inst(i+1, r) \defeq T(\inst_1, \inst_2)$. Here, $\inst_1 \defeq \inst(i, 3r)$ and $\inst_2 \defeq \inst(i, u_{i, 3r})$, with parameters $(a_{1}, \ldots, u_{1})$ and $(a_2, \ldots, u_2)$, respectively.
\end{Def}

In other words, $\inst_1$ is defined such that $r_1 = 3 r_+$ and $\inst_2$ such that $r_2 = u_1$. In \Cref{sec:lower-bounds-appendix} we study the growth of the parameters of our gap instance families. Two parameters that are easy to bound for this construction are the following.

\begin{obs}\label{tensoring-routing-coding}
	For any $i\geq 0$ and $r\geq 5$, we have $a_{i,r}=3^{2^{i}}$ and $b_{i,r}=5^{2^{i}}$.
\end{obs}

A less immediate bound, whose proof is also deferred to \Cref{sec:lower-bounds-appendix}, is the following bound on the number of edges of the gap instances..

\begin{restatable}{lem}{mbound}\label{lem:iterated-m-bound}
  We have that $\log m_{i, r} \le 2^{O(2^i)} \log r$ for all $i \ge 0, r \ge 5$.
\end{restatable}

\subsection{Lower Bounding the Coding Gap}

We are now ready to prove this section's main result -- a polylog$(k)$ makespan coding gap.
\latencylb*
\begin{proof}
  For each $i \ge 0$ and $r \defeq 5^{2^i}$, consider $\inst_{i,r}$ as defined above. By \Cref{tensoring-routing-coding} this gap instance has coding makespan at most $a_{i, r} = 3^{2^i}$. Moreover, also by \Cref{tensoring-routing-coding}, this instance has $b_{i, r} = 5^{2^i}$, and so by  \Cref{lem:parameters-to-routing-lb} its routing makespan is at least $5^{2^i}$. Hence the makespan coding gap of $\inst_{i,r}$ is at least $(5/3)^{2^i}$. It remains to bound this gap in terms of $k\defeq k_{i.r}$.
  
  As the terminals of $\inst_{i,r}$ are disjoint, we have that $k$ is upper bounded by the number of nodes of $\inst_{i,r}$, which is in turn upper bounded by $m_{i,r}$, as $\inst_{i,r}$ is connected and not acyclic.
  That is, $k\leq m_{i,r}$.
  But by \Cref{lem:iterated-m-bound}, we have that $\log m_{i, r} \le 2^{O(2^i)} \cdot \log r = 2^{O(2^i)} \cdot O(2^i) \leq 2^{O(2^i)} \le 2^{c' \cdot 2^i}$,
  for some universal constant $c' > 0$.
  Therefore, stated in terms of $k$, the makespan coding gap is at least
  \begin{align*}
  (5/3)^{2^i} & = 2^{2^i \log 5/3} = ( 2^{c' 2^i} )^{\frac{\log 5/3}{c'}} \ge (\log m_{i, r})^{c} \ge \log^{c} k,
  \end{align*}
  where $c \defeq \frac{\log 5/3}{c'}>0$ is a universal constant, as claimed.
\end{proof}


\section{Coding Gaps for Other Functions of Completion Times}\label{sec:lp-norm-gaps}

In this section we extend our coding gap results to other time complexity measures besides the makespan. For our upper bounds, we show that our coding gaps for $\ell_\infty$ minimization of the completion times (makespan) implies similar bounds for a wide variety of functions, including all weighted $\ell_p$ norms; proving in a sense that $\ell_{\infty}$ is the ``hardest'' norm to bound. The following lemma underlies this connection.

\begin{lem}\label{ell-inf-domination}
	Let $\alpha$ be an upper bound on the coding gap for completion times' $\ell_\infty$ norm (makespan). 
	Then, if multiple-unicast instance $\calM$ admits a coding protocol with completion times $(T_1,T_2,\dots,T_k)$, there exists a routing protocol for $\calM$ with completion times placewise at most $(4\alpha\cdot T_1,4\alpha\cdot T_2,\dots,4\alpha\cdot T_k)$.
\end{lem}
\begin{proof}
	Let $\calM$ be a multiple-unicast instance. 
	Let $(T_1,T_2,\dots,T_k)$ be the vector of completion times of some coding protocol. Without loss of generality, assume $T_1\leq T_2\leq \dots \leq T_k$. Next, for any $j\in \mathbb{Z}$, denote by $\calM_j$ the sub-instance of $\calM$ induced by the unicasts with completion time $T_i \in [2^j, 2^{j+1})$. Then, there exists a network coding protocol for each $\calM_j$ with makespan at most $2^{j+1}$. Consequently, there exists a routing protocol for $\calM_j$ with makespan at most $\alpha\cdot 2^{j+1}$. Scheduling these protocols in parallel, in order of increasing $j=0,1,2,\dots,$ we find that a unicast with completion time $T_i\in [2^j,2^{j+1})$ in the optimal coding protocol has completion time in the obtained routing protocol which is at most
	 \begin{align*}\sum_{j'\leq j} \alpha\cdot 2^{j'+1} & \leq 2\alpha \cdot 2^{j+1} \leq 4\alpha\cdot T_i.\qedhere
	\end{align*}
\end{proof}
Note that unlike our routing protocols for makespan minimization of \Cref{wc-gap-ub}, the proof here is non-constructive, as it assumes (approximate) knowledge of the completion times of each unicast in the optimal coding protocol. Nonetheless, this proof guarantees the existence of a protocol, which suffices for our needs. In particular, 
applying \Cref{ell-inf-domination} to the coding protocol minimizing a given weighted $\ell_p$ norm, we immediately obtain the following.
\begin{cor}
	Let $\alpha$ be an upper bound on the coding gap for completion times' $\ell_\infty$ norm (makespan). Then the coding gap for any weighted $\ell_p$ norm of the completion times is at most $4\alpha$.
\end{cor}
Plugging in our coding gap upper bound of \Cref{wc-gap-ub}, we therefore obtain a generalization of \Cref{wc-gap-ub} to any weighted $\ell_p$ norm, as well as average completion time (which corresponds to $\ell_1$).
\begin{thm}\label{thm:ell_p_ub}
	The network coding gap for any weighted $\ell_p$ norms of completion times is at most
	$$O\left(\log(k) \cdot \log \left(\sum_i d_i / \min_i d_i\right)\right).$$
\end{thm}
Note that similar bounds hold even more generally. In particular, for any sub-homogeneous function of degree $d$ (i.e., $f(c\cdot \vec{x})\leq c^d\cdot f(\vec{x})$, \Cref{ell-inf-domination} implies a coding gap of at most $(4\alpha)^d$, where $\alpha$ is the best upper bound on the coding gap for makespan minimization. 

\paragraph{Lower bounds.} 
As with makespan minimization, a polylogarithmic dependence in the problem parameters as in \Cref{thm:ell_p_ub}, as we prove below. Crucially, we rely on our makespan coding gap's examples displaying the property that under coding nearly all unicast sessions' completion time is at least polylogarithmically larger than under the best coding protocol.

\begin{thm}
  There exists an absolute constant $c > 0$ and an infinite family of $k$-unicast instances whose $\ell_p$-coding gap is at least $\Omega(\log^{c} k)$.
\end{thm}
\begin{proof}
  We follow the proof of \Cref{wc-gap-lb} and consider $I_{i, r}$, this time setting $r \defeq (5^{2^i})^2 = 5^{2^{i+1}}$. This does not change the parameters $a_{i, r} = 3^{2^i}$ and $b_{i, r} = 5^{2^i}$, nor the bound $\alpha \defeq (5/3)^{2^i} \ge \log^c k$ for some absolute constant $c > 0$.
  This instance has a coding protocol with completion times $(a_{i, r}, \ldots, a_{i, r})$, and so this coding protocol's completion times' $\ell_p$ value is $a_{i,r}$. On the other hand, 
  by \Cref{lem:parameters-to-routing-lb}, at least $k \cdot (1 - \frac{b_{i, r}}{r}) \ge (1 - o(1)) \cdot k$ pairs have routing completion time at least $b_{i, r}$, where $o(1)$ tends to $0$ as $i \to \infty$. Consequently,
  the $\ell_p$-value of any routing protocol's completion times is at least $(1 - o(1)) \cdot b_{i, r}$. Since $b_{i,r} / a_{i, r} = \alpha \ge \log^c k$, we obtain the required coding gap.
\end{proof}


\section{Conclusions and Open Questions}\label{sec:conclusion}

In this paper we study completion-time coding gaps; i.e., the ratio for a given multiple-unicast instance, of the fastest routing protocol's completion time to the fastest coding protocol's completion time. We provide a strong characterization of these gaps in the worst case, showing they can be polylogarithmic in the problem parameters, but no greater. The paper raises a few exciting questions and research directions.

\smallskip

Probably the most natural question is to close our upper and lower bounds. We show that the network coding gap is polylogarithmic, but what polylog?
Another question, motivated by the super-constant speedups we prove coding can achieve over routing for this basic communication problem, is whether there exist efficient algorithms to compute the fastest coding protocol, mirroring results known for the fastest routing protocols for multiple unicasts \cite{leighton1994packet,leighton1999fast}, and for the highest-throughput coding protocols for multicast \cite{jaggi2003low}.
Another natural question is extending this study of network coding gaps for completion times to other widely-studied communication problems, such as multiple multicasts. 



\paragraph{Implications to other fields.} 
As discussed in \Cref{sec:intro} and \Cref{sec:related}, the conjectured non-existence of throughput coding gaps for multiple unicast has been used to prove (conditional) lower bounds in many seemingly-unrelated problems. 
It would be interesting to see whether our upper and lower bounds on the coding gap for multiple unicasts' completion times can be used to prove \emph{unconditional} lower bounds for other models of computation. 
We already have reason to believe as much;
using techniques developed in this paper, combined with many other ideas, the authors have obtained the first non-trivial universal lower bounds in distributed computation. It would be interesting to see what other implications this work might have to other areas of Theory.
Perhaps most exciting would be to investigate whether completion-time coding gaps imply new results in circuit complexity for depth-bounded circuits.

\paragraph{Acknowledgements}

The authors would like to thank Mohsen Ghaffari for suggesting an improvement to \Cref{wc-gap-ub} which resulted in a coding gap independent of $n$, Anupam Gupta for pointing out a simplification of \Cref{lemMetricAllToAll} and the Lemma's similarity to \cite[Theorem 1]{arora2009expander}, and Paritosh Garg for bringing \cite{braverman2017coding} to our attention.

\appendix
\section*{Appendix}
\section{Completion Time vs. Throughput}\label{sec:completion-vs-throughput}

In this section we argue why network coding upper bounds for makespan imply coding gaps for throughput maximization. We first introduce the standard definitions of the throughput maximization model~\cite{adler2006capacity, braverman2017coding}. The differences between the throughput and completion-time model (see~\Cref{sec:prelims}) are {\color{blue}highlighted in blue}.

\textbf{Throughput maximization model.} A \emph{multiple-unicast instance} $\calM=(G,\calS)$ is defined over a communication network, represented by an undirected graph $G = (V,E,c)$ with capacity $c_e\in \integers_{\ge 1}$ for each edge $e$. The $k\defeq|\calS|$ \emph{sessions} of $\calM$ are denoted by $\calS = \{ (s_i, t_i, d_i) \}_{i=1}^k$. {\color{blue}The (maximum) \emph{throughput} of $\calM$ is the supremum $r > 0$ such that there exists a sufficiently large $b > 0$ where the following problem has a correct protocol.}
Each session consists of source node $s_i$, which wants to transmit a packet to its sink $t_i$, consisting of ${\color{blue}\lceil r \cdot b \cdot d_i \rceil}$ sub-packets (e.g., an element of an underlying field). A \emph{protocol} for a multiple-unicast instance is conducted over finitely-many \emph{synchronous time steps}. Initially, each source $s_i$ knows its packet, consisting of $d_i$ sub-packets.
At any time step, the protocol instructs each node $v$ to send a different packet along each of its edges $e$. The packet contents are computed with some predetermined function of packets received in prior rounds by $v$ or originating at $v$.
{\color{blue}The total number of sub-packets sent through an edge $e$ over the duration of the entire protocol is at most $b\cdot c_e$. We differentiate the maximum throughput achievable by coding and routing protocols as $r^{C}$ and $r^{R}$, respectively. The throughput coding gap is the largest ratio $r^{C} / r^{R}$ that can be achieved for any instance.}

\textbf{Relating completion times and throughput.} The throughput maximization intuitively corresponds to the makespan minimization of an instance with asymptotically-large packet sizes. More formally, we modify a multiple unicast instance $\calM$ by increasing its demands by a factor of $w$ while keeping the capacities the same. This causes the makespan to increase. We argue that the slope of the increase with respect to $w$ is exactly the throughput of $\calM$.

\begin{Def}\label{def:Cw}
  Given a multiple-unicast instance $\calM$ we define $C^{C}(w)$ and $C^{R}(w)$ to be the makespan of the fastest coding and routing protocols when the all demands are multiplied by a common factor $w$.
\end{Def}

\begin{obs}\label{obs:makespan-to-throughput}
  Let $\calM$ be a multiple-unicast instance. The maximum throughput $r$ corresponding to $\calM$ is equal to $\sup_{w \to \infty} w / C(w)$ for both coding and routing. Formally, $r^C = \sup_{w \to \infty} w / C^C(w)$ and $r^R = \sup_{w \to \infty} w / C^R(w)$.
\end{obs}
\begin{proof}
  We drop the R/C superscripts since the proof holds for both without modification. Let $L \defeq \sup_{w \to \infty} w / C(w)$. We first argue that $L \ge r$, i.e., we can convert a throughput protocol to a makespan-bounded one. For simplicity, we will assume that $b = 1$ ($b$ from the throughput definition); when this is not the case one needs to appropriately re-scale the sub-packets for the completion-time protocol.

  Let $\calT$ be a protocol of throughput at least $r - o(1)$ and let $T$ be the total number of rounds $\calT$ uses (note that in the throughput setting $T$ has no impact on the quality of $\calT$). Let $w \in \mathbb{Z}$ be a sufficiently large number. We use pipelining by scheduling $w' \defeq w / (r - o(1))$ independent copies of $\calT$: the first one starting at time $1$, second at time $2$, ..., last one at time $w'$. Each copy operates on a separate set of sub-packets, with the pipelined protocol being able to transmit $(r - o(1)) \cdot d_i \cdot w' = d_i \cdot w$ sub-packets across the network (in line with \Cref{def:Cw}) in at most $w' + T$ rounds. Note that $\calT$ sends at most $c_e$ sub-packets over an edge $e$ over its entire execution, hence the pipelined version of $\calT$ never sends more than $c_e$ sub-packets during any one round. In other words, we have that $C(w) \le w' + T$. Letting $w \to \infty$ (which implies $w' \to \infty$), we have that
  \begin{align*}
    L \ge \frac{w}{C(w)} \ge \frac{w}{w' + T} = (r - o(1))\frac{w'}{w' + T} = r - o(1).
  \end{align*}

  We now argue the converse $r \ge L$, i.e., we can convert a completion-time protocol into a throughput protocol with the appropriate rate. The result essentially follows by definition. By assumption, for some sufficiently large $w > 0$ there exists a protocol with makespan at most $C(w) \le w / (L - o(1))$. The protocol sends a total of at most $C(w) c_e$ sub-packets over an edge $e$. Furthermore, by construction of $C(w)$, each source-sink pair successfully transmits $w \cdot d_i$ sub-packets. By noting that $w \cdot d_i = (L - o(1)) C(w) d_i$, we conclude that by using $b \defeq C(w)$ we get a protocol with rate $L - o(1)$.
\end{proof}

\begin{cor}
  Suppose that the makespan coding gap is at most $\alpha$ (over all instances). Then the throughput coding gap is at most $\alpha$.
\end{cor}
\begin{proof}
  Consider some multiple unicast instance $\calM$, with coding throughput $r$.
%
  By \Cref{obs:makespan-to-throughput}, for sufficiently large $w$ there is a coding protocol $\calP_1$ satisfying $w / C^C(w) \ge r - o(1)$, i.e., $C^C(w) \le w / (r - o(1))$. By the makespan coding gap assumption, there exists a routing protocol $\calP_2$ implying that $C^R(w) \le \alpha \cdot w / (r - o(1))$. Furthermore, following the proof of \Cref{obs:makespan-to-throughput}, protocol $\calP_2$ implies a routing throughput of $r'$ for the original instance $\calM$, satisfying
  \begin{align*}
    r' \ge w / C^R(w) \ge (r - o(1)) / \alpha = r / \alpha - o(1).
  \end{align*}
  In other words, $r / r' \ge \alpha + o(1)$ and we are done.
\end{proof}


\section{Network Coding Model for Completion Time}\label{sec:network-coding-model}

In this section we formalize the $k$-session unicast communication problem and the notion of completion time for it. We note that our model is not new---e.g., it is equivalent to the models of \citet{chekuri2015delay, wang2016sending} that study delay in communication networks.

The input to a \emph{$k$-session unicast problem} $(G, \calS)$ consists of a graph $G = (V, E, c)$, where edges have capacities $c : E \to \mathbb{R}_{\ge 0}$, and a set of $k$ sessions $\calS = \{(s_i, t_i, d_i)\}_{i=1}^k$. Each triplet $(s_i, t_i, d_i)$ corresponds to the source $s_i \in V$, sink $t_i \in V$ and the demand $d_i \in \mathbb{R}_{\ge 0}$ of session $i$. The graph $G$ can be either directed or undirected, where in the latter case we model an undirected edge $e$ as two directed edges $\vec{e}, \cev{e}$ where both of them have equal capacity $c(\vec{e}) = c(\cev{e}) = c_e$.\footnote{Papers such as Adler et al.~\cite{adler2006capacity} often impose an alternative condition $c(\vec{e}) + c(\cev{e}) = c(e)$, which would make our proofs slightly heavier on notation. However, their convention can only impact the results up to a factor of 2, which we typically ignore in this paper.}

Each source $s_i$ is privy to an input message $m_i \in M_i$ generated by an arbitrary stochastic source with entropy at least $d_i$, hence the entropy of the random variable $m_i$ is $d_i$. The sources corresponding to different sessions are independent.

A $T$-round \emph{network coding computation} consists of a set of $|E| \times T$ coding functions $\{f_{\vec{e}, r} : M \to \Gamma\}_{\vec{e} \in E, 1 \le r \le T}$, where $\Gamma$ is some arbitrary alphabet and $M \defeq \prod_i M_i$. These functions satisfy the following properties:
\begin{itemize}
\item The entropy of any coding function $f_{\vec{e}, r}$ never exceeds the edge capacity $c(\vec{e})$, i.e., $H(f_{\vec{e}, r}) \le c(\vec{e})$ for all $\vec{e} \in E, 1 \le r \le T$.
\item For each directed edge $\vec{e} = (u, v) \in E$ and round $1 \le r \le T$ the function $f_{\vec{e}, r}$ is computable from communication history received strictly before round $r$ at node $u$. In other words, let the communication history $Y_{u, r}$ be defined as $\{ m_i \mid i \in [k], s_i = u \} \cup \{ f_{(x, y), r'} \mid y = u \text{ and } r' < r \}$, then $H(f_{(u, v), r} | Y_{u, r}) = 0$.
\item The completion times of a network coding computation are $(T_1, T_2, \ldots, T_k) \in \mathbb{Z}_{\ge 0}^k$ when the following holds. For every session $i$, the message $m_i$ of the session $(s_i, t_i, d_i)$ must be computable from the sink $t_i$'s history after $T_i$ rounds are executed, i.e., $H(m_i | Y_{t_i, T_i+1}) = 0$.
\end{itemize}

\textbf{Remark:} The above ``bare-bones'' formalization is sufficient for all of our results to hold. However, such a formalization can be unwieldy since a complete instance description would also need to specify a stochastic distribution corresponding to each source $s_i$. A standard way of avoiding this issue is to simply assume the sources generate a uniformly random binary string of length $d_i$ (forcing $d_i$ to be an integer). Without going into too much detail, we mention this assumption can be made without loss of generality if we allow for (1) an arbitrarily small decoding error $\eps > 0$, (2) slightly perturbing the edge capacities $c_e$ and source entropies $d_i$ by $\eps$, and (3) scaling-up both $c_e$'s and $d_i$'s by a common constant $b > 0$; this approach is standard in the literature (e.g., see~\cite{adler2006capacity, braverman2017coding}).


\section{Deferred Proofs of \Cref{sec:upper-bounds}}\label{sec:upper-bounds-appendix}

In this section we provide proofs deferred from \Cref{sec:upper-bounds}, starting with the proof of \Cref{lem:primal-to-routing-protocol}, restated here for ease of reference.

\primalroutingprotocol*
To prove the above, we rely on the celebrated $O(\textrm{congestion + dilation})$ packet scheduling theorem of \citet{leighton1994packet}. In particular, 
we use the solution to \primal$(T)$ to obtain a collection of short paths with bounded congestion (i.e., bounded maximum number of paths any given edge belongs to). We then route along these paths in time proportional to these paths' maximum length and congestion.
The issue is that the feasible LP solution provides fractional paths, hence requiring us to round the LP solution. Independent rounding would result in paths of length $T$ and congestion $T/z + O(\log n)$ (with high probability). To avoid this additive dependence on $n$, we rely on the following theorem of \citet{srinivasan2001constant}.

\begin{lem}[\cite{srinivasan2001constant}, Theorem 2.4, paraphrased]\label{rounding-routing}
	Let $\calM$ be a multiple-unicast instance and for each $i \in [k]$ let $\calD_i$ be a distribution over $s_i \pathto t_i$ paths of hop-length at most $L$. Suppose that the product distribution $\prod \calD_i$ has expected congestion for each edge at most $L$. Then there exists a sample $\omega \in \prod \calD_i$ (i.e., a choice of a from $\calD_i$ between each $s_i \pathto t_i$) with (maximum) congestion $O(L)$.
\end{lem}

Using the above lemma to round the LP and  using \citet{leighton1994packet} path routing to route along the obtained paths yields \Cref{lem:primal-to-routing-protocol}.
\begin{proof}[Proof of \Cref{lem:primal-to-routing-protocol}]
	Consider an optimal solution to this \primal($T$). 
	Clearly, picking for each pair $(s_i,t_i)$ some $d_i$ paths in $P_i(T)$ with each $p\in \calP_i(T)$ picked with probability $f_i(p)\cdot d_i/\sum_{p\in \calP_i(T)} f_i(p) \leq f_i(p)/z$ yields an expected congestion at most $T\cdot c_e / z$ for each edge $e$. That is, thinking of $G$ as a multigraph with $c_e$ copies per edge, each such parallel edge has congestion $T/z$. On the other hand, each such path has length at most $T\leq T/z$ (since $z\leq 1$). Therefore, by \Cref{rounding-routing}, there exist choices of paths for each pair of (maximum) congestion and hop-bound (i.e., dilation) at most $O(T/z)$. But then, using $O(\mbox{congestion + dilation)}$ routing \cite{leighton1994packet} this implies an integral routing protocol with makespan $O(T/z)$, as claimed.
\end{proof}

Here we prove \Cref{lem:bucketing-lemma}, restated here for ease of reference.

\bucketing*
\begin{proof}
	Suppose (without loss of generality) that $h_1 \ge h_2 \ge \ldots h_k$ and assume for the sake of contradiction that none of the sets $[1], [2], \ldots, [k]$ satisfy the condition. In other words, if we let $d([j]) \defeq \sum_{i = 1}^j d_j$, then $h_i < \frac{1}{\alpha} \cdot \frac{1}{d([i])}$ for all $i \in [k]$. Multiplying both sides by $d_i$ and summing them up, we get that $1 \le \sum_{i=1}^k d_i h_i < \frac{1}{\alpha} \sum_{i=1}^k \frac{d_i}{d([i])}$. Reordering terms, this implies $\sum_{i=1}^k \frac{d_i}{d([i])} > \alpha$.
	
	Define $f(x)$ as $1/d_1$ on $[0, d_1)$; $1/(d_1+d_2)$ on $[d_1, d_1+d_2)$; ...; $1/d([i])$ on $[d([i-1]), d([i]))$ for $i \in [k]$. Now we have
	\begin{align*}
	\int_0^{d([k])} f(x) = \frac{d_1}{d_1} + \frac{d_2}{d_1 + d_2} + \frac{d_3}{d_1 + d_2 + d_3} + \ldots + \frac{d_k}{d([k])} .
	\end{align*}
	However, since $f(x) \le 1/x$
	\begin{align*}
	\int_0^{d([k])} f(x) & = \int_0^{d_1} f(x)\, dx + \int_{d_1}^{d([k])} f(x)\, dx \\
	& \le 1 + \int_{d_1}^{d([k])} \frac{1}{x}\, dx = 1 + \ln \frac{d([k])}{d_1} .
	\end{align*}
	Hence by setting $\alpha \defeq 1 + \ln \frac{d([k])}{d_1}$ we reach a contradiction and finish the proof.
\end{proof}


\section{Deferred Proofs of \Cref{sec:lower-bounds}}\label{sec:lower-bounds-appendix}

Here we provide the deferred proofs of lemmas of \Cref{sec:lower-bounds}, restated below for ease of reference.

\tensoredcodingtime*
\begin{proof}
	Suppose there exists a network coding protocol with makespan $t_i \le a_i$ that solves $(G_i, \calS_i)$ for $i \in \{1, 2\}$. Functionally, each round in the outer instance $(G_1, \calS_1)$ consists of transmitting $c_e$ bits of data from $u$ to $v$ for all arcs $(u, v)$ where $\{u, v\} \in E(G_1)$. This is achieved by running the full $t_2$ rounds of the inner instance protocol over all copies of the instances which pushes $d_i$ bits from $s_i$ to $t_i$ for all $(s_i, t_i, d_i)$ and all copies of the inner instance. The reason why such inner protocol pushes the information across each arc $(u, v)$ is because $u$ is merged with some $s_i$, $v$ is merged with $t_i$, and with $d_i = c_e$ for some $(s_i, t_i, d_i) \in \calS_2$ and some copy of the inner instance. In conclusion, by running $t_1$ outer rounds, each consisting of $t_2$ inner rounds, we get a $t_1 t_2 \le a_1 a_2$ round protocol for the product instance.
\end{proof}

\tensorparams*
\begin{proof}[Proof of \Cref{lem:one-tensor-params}]
	First, the set of terminals in the product instance $\inst_+$ is disjoint, as distinct terminals of copies of the outer instance $\inst_1$ have their edges associated with distinct terminals source-sink pairs of the inner instance $\inst_2$. Consequently, no two terminals of the outer instance are associated with the same node of the same copy of an inner instance. The capacities and demands of $\inst_+$ are one by definition. We now turn to bounding the gap instance's parameters.
	
	Parameters $a_+$ and $b_+$ are directly argued by \Cref{lem:tensoring-coding} and \Cref{lem:tensoring-routing}. Furthermore, $f_+, k_+, m_+$ are obtained by direct counting, as follows. 
	
	Recall that the cut edges of the outer instance get replaced with a path of length $a_2$. Since there are $n_1$ copies of outer instances, each having $f_1$ cut edges, this contributes $a_2 n_1 f_1$ edges to $m_+$. The non-cut edges of the outer instance get deleted and serve as a merging directive, hence they do not contribute to $m_+$. Finally, each edge of the inner instance gets copied into $\inst_+$, contributing $n_2 m_2$ as there are $n_2$ copies of the inner instance.
	
	For $r_+$ we need to show it is a lower bound on $k_+/f_+$. We note that $|E(B)| = n_1 \cdot 2(m_1 - f_1) = n_2 k_2$ and proceed by direct calculation:
	\begin{align*}
	\frac{k_+}{f_+} & = \frac{n_1 k_1}{n_1 f_1 + n_2 f_2} = \frac{k_1}{f_1} \cdot \frac{1}{1 + \frac{n_2}{n_1} \frac{f_2}{f_1}} \ge \frac{k_1}{f_1} \cdot \frac{1}{1 + \frac{2(m_1 - f_1)}{k_2} \frac{f_2}{f_1}} \\
	& \ge \frac{k_1}{f_1} \cdot \frac{1}{1 + \frac{2m_1}{f_1} \frac{f_2}{k_2}} \ge \frac{k_1}{f_1} \cdot \frac{1}{1 + \frac{2 u_1}{r_2}} = r_1 \cdot \frac{1}{1 + \frac{2 u_1}{r_2}} = r_+.
	\end{align*}
	
	For $u_+$ we need to show it is an upper bound on $m_+/f_+$. Note that $k_2 \le m_2$ since the set of terminals is disjoint and the graph is connected.
	\begin{align*}
	\frac{m_+}{f_+} & = \frac{n_2 m_2 + a_2 n_1 f_1}{n_2 f_2 + n_1 f_1} = \frac{m_2}{f_2} \cdot \frac{1 + a_2 \frac{n_1}{n_2} \frac{f_1}{m_2}}{1 + \frac{n_1}{n_2} \frac{f_1}{f_2}} \le u_2 \cdot \frac{1 + a_2 \frac{k_2}{2(m_1-f_1)} \frac{f_1}{m_2}}{1 + \frac{k_2}{2(m_1-f_1)} \frac{f_1}{f_2}}\\
	& \le u_2 \cdot \frac{1 + a_2 \frac{k_2}{2(m_1/f_1 - 1)} \frac{1}{m_2}}{1 + \frac{f_1}{2 m_1} \frac{k_2}{f_2}} \le u_2 \cdot \frac{1 + a_2 \cdot \frac{1}{2} \cdot 1}{1 + \frac{1}{2} \frac{r_2}{u_1}}
	= u_+.
	\end{align*}
	Here the last inequality relies on $m_1/f_1\geq 2$ and on $k_2\leq m_2$, which follows from the set of terminals being disjoint and the graph $G_2$ being connected.
	
	For the final technical conditions, note that $a_+ \ge 2$ is clear from $a_+ = a_1 a_2 \ge 4 \ge 2$. Finally, $\frac{m_+}{f_+} \ge 2$ follows from the following.
	\begin{align*}
	\frac{m_+}{f_+} & = \frac{a_2 n_1 f_1 + n_2 m_2}{n_1 f_1 + n_2 f_2} = a_2 \frac{n_1 f_1}{n_1 f_1 + n_2 f_2} + \frac{m_2}{f_2} \frac{n_2 f_2}{n_1 f_1 + n_2 f_2} \ge 2  \left(\frac{n_1 f_1}{n_1 f_1 + n_2 f_2} + \frac{n_2 f_2}{n_1 f_1 + n_2 f_2}\right) = 2.\qedhere
	\end{align*}
\end{proof}

\subsection{Upper Bounding $m_{i,r}$}

For readability, we sometimes write $u(i, r)$ instead of $u_{i, r}$ and similarly for $m(i, r)$. Also, we note that the technical conditions $a_{i, r} \ge 2$ and $\frac{m_{i, r}}{f_{i, r}} \ge 2$, which clearly hold for $i=0$, hold for all $i>0$, due to \Cref{lem:one-tensor-params}. Finally, we note that by \Cref{lem:one-tensor-params}, since $r_2=u_1$ and $a_2\geq 1$, we have that $u_+ \geq u_2$ and so for all gap instances in the family we have $u_{i,r}\geq u_{i-1,u(i-1,3r)} \geq 5$.

\begin{lem}\label{lem:iterated-tensoring-invariants}
	The parameter of $\inst(i, r)$ for any $i \ge 0, r \ge 5$ satisfy the following. 
	\begin{enumerate}
		[(i)]
		\item \label{bound:r} $\frac{k_{i, r}}{f_{i, r}} \ge r$,
		\item \label{bound:u} $u_{i+1, r} \le 3^{2^i} \cdot u_{i, u(i, 3r)}$ and
		\item \label{bound:m} $\log m(i+1, r) \le O(5^{2^{i+1}}) \cdot \log (m_{i, 3r} \cdot m_{i, u(i, 3r)})$.
	\end{enumerate}
\end{lem}
\begin{proof}
	Claim \ref{bound:r} follows from \Cref{lem:one-tensor-params}, as follows.
	\begin{align*}
	\frac{k_{i+1,r}}{f_{i+1,r}} \ge r_{i+1, r} = r_{i, 3r} \cdot \frac{1}{1 + 2u_{i, 3r} / r_{i, u(i, 3r)}} \ge 3r \cdot \frac{1}{1 + 2u_{i, 3r} / u_{i, 3r}} = 3r \cdot \frac{1}{3} = r.
	\end{align*}
	
	We now prove claims \ref{bound:u} and \ref{bound:m}. Fix $i, r$ and define $\inst_1 \defeq \inst(i, 3r)$ (with parameters $(a_{1}, \ldots, u_{1})$) and $\inst_2 \defeq \inst(i, u_{i, r})$ (with parameters $(a_{2}, \ldots, u_{2})$). We have $u(i+1, r) = u_2 \frac{1 + a_2/2}{1 + r_2 / (2 u_1)} \le u_2 \frac{1 + a_2/2}{1 + 1/2} \le u_2 \cdot a_2$ (\Cref{lem:one-tensor-params}), with $a_2 = a_{i, u(i,3r)} = 3^{2^i}$ and $u_2 = u(i, u(i, 3r))$ from the iterated tensoring process. Therefore, we conclude that $u(i+1, r) \leq 3^{2^i} \cdot u_{i, u(i, 3r)}$, as claimed.
	
	We now prove Claim \ref{bound:m}. The corresponding colored bipartite graph $B \in \calB(n_1, n_2, 2(m_{1} - f_{1}), k_{2}, 2 b_{1} b_{2} )$ used to produce the product $\inst_{i, r}$ has $\max(n_1, n_2) \le ( 2(m_{1} - f_{1}) k_{2})^{O(b_{1} b_{2})}$, by \Cref{colored-bip}. Therefore, as $k_2\leq m_2$, we have that $\max(n_1, n_2) \le (m_{1} \cdot m_{2})^{O(b_{1} b_{2})}$. This implies the following recurrence for $m_{i,r}$.
	\begin{align*}
	m_{i+1, r} & = a_2 n_1 f_1 + n_2 m_2 \le a_2 \max(n_1, n_2) m_1 m_2.
	\end{align*}
	Taking out logs, we obtain the desired bound.
	\begin{align*}
	\log m_{i+1, r} & \le \log a_2 + \log \max(n_1, n_2) + \log m_1 m_2 \\
	& = O(2^i) + O(b_1 b_2) \log (m_1 m_2) + \log (m_1 m_2) \\
	& = O(2^i) + O(5^{2^i}) \log (m_1 m_2) \\
	& = O(5^{2^i}) \cdot \log (m_{i, 3r} \cdot m_{i, u(i, 3r)}).\qedhere
	\end{align*}
\end{proof}

Given \Cref{lem:iterated-tensoring-invariants} we obtain the bound on $u_{i,r}$ in terms of $i$ and $r$. 
\begin{restatable}{lem}{ubound}\label{lem:iterated-u-bound}
	We have that $\log u_{i, r} \le 2^{O(2^i)} \log r$ for all $i \ge 0, r \ge 5$.
\end{restatable}
	\begin{proof}
		By \Cref{lem:iterated-tensoring-invariants}, we have the recursion $u(i+1, r) \le 3^{2^i} \cdot u(i, u(i, 3r))$ with initial condition $u(0, r) = O(r^2)$, by \Cref{base-instance}. Taking out logs, we obtain $\log u(i+1, r) \le O(2^i) + \log u(i, u(i, 3r))$ and $\log u(0,r) = O(\log r)$.
		We prove via induction that $\log u(i, r) \le \frac{1}{c} (c^2)^{2^i} \cdot \log r$ for some sufficiently large $c > 0$. In the base case $\log u(0, r) = O(\log r) \le \frac{1}{c} (c^2) \log r = c \log r$. For the inductive step we have:
		\begin{align*}
		\log u(i+1, r) & \le O(2^i) + \log u(i, u(i, 3r)) \\
		& \le O(2^i) + \frac{1}{c} (c^2)^{2^{i}} \log u(i, 3r) \\
		& \le O(2^i) + \frac{1}{c} (c^2)^{2^{i}} \frac{1}{c} c^{2^{i}} \log 3r \\
		& \le O(2^i) + \frac{1}{c^2} (c^2)^{2^{i+1}} \log 3r \\
		& \le \frac{1}{c} (c^2)^{2^{i+1}} \log r,
		\end{align*}
		where the last inequality holds for $i\geq 1$ and $r\geq 5$ and a sufficiently large $c > 0$.
	\end{proof}

	Plugging in the bound of \Cref{lem:iterated-u-bound} and \Cref{lem:iterated-tensoring-invariants} we can prove inductively the upper bound on the number of edges of $\inst_{i,r}$ in terms of $i$ and $r$ given by \Cref{lem:iterated-m-bound}, restated here.
	\mbound*
	\begin{proof}
		By \Cref{lem:iterated-tensoring-invariants}, we have the recursion $\log m(i+1, r) \le O(5^{2^{i+1}}) \cdot \log (m_{i, 3r} \cdot m_{i, u(i, 3r)})$ with initial condition $\log m(0, r) = O(\log r)$, by \Cref{base-instance}. We prove via induction that $\log m_{i, r} \le c^{2^i} \log r$ for a sufficiently large universal constant $c > 0$. In the base case $\log m_{0, r} \le O(\log r) \le c \log r = c^{2^0} \log r$. For the inductive step, using \Cref{lem:iterated-u-bound} to bound $\log u(i,3r)$, we have:
		\begin{align*}
		\log m_{i+1, r} & \le O(5^{2^{i+1}}) \cdot (\log m_{i, 3r} + \log  m_{i, u(i, 3r)}) \\
		& \le O(5^{2^{i+1}}) \cdot \left(c^{2^i} \log 3r + c^{2^i} \log u(i, 3r) \right) \\
		& \le O(5^{2^{i+1}}) \cdot \left(c^{2^i} \log 3r + c^{2^i} 2^{O(2^i)} \log 3r \right) \\
		& = c^{2^i}\cdot O(5^{2^{i+1}})\cdot 2^{O(2^i)} \log 3r \\
		& \le c^{2^i+1} \log r,
		\end{align*}
		where the last inequality holds for $i\geq 1$ and $r\geq 5$ and a sufficiently large $c > 0$.
	\end{proof}


\bibliographystyle{acmsmall}
\bibliography{abb,network-coding,ultimate}

\end{document}